\documentclass[11pt,a4paper]{article}
\usepackage{graphicx,nicefrac}
\usepackage{tabularx,booktabs,multirow}
\usepackage{todonotes}
\usepackage{amssymb,amsmath,amsthm}
\usepackage{mathtools}
\usepackage[pagebackref]{hyperref}
\usepackage{cleveref}
\usepackage{subcaption}
\usepackage{amsmath,amsfonts}
\usepackage{xcolor}
\usepackage{xspace}
\usepackage{bbm}
\usepackage{tikz}
\usepackage{makecell}
\usepackage{authblk}
\usepackage{fullpage}

\usepackage{booktabs}
\usepackage{mathrsfs}
\usepackage[square,numbers]{natbib}
\usepackage{graphicx,dsfont}

\usepackage{placeins}
\usepackage{multicol}
\usepackage{multirow}
\usepackage{caption}
\usepackage{comment}
\usepackage{amsfonts}
\usepackage{algorithm}
\usepackage[noend]{algpseudocode}
\algnewcommand\algorithmicinput{\textbf{Input:}}
\algnewcommand\algorithmicoutput{\textbf{Output:}}
\algnewcommand\Input{\item[\algorithmicinput]}
\algnewcommand\Output{\item[\algorithmicoutput]}
\algnewcommand\algorithmicgoto{\textbf{GoTo }}
\algnewcommand\GoTo{\item[\algorithmicgoto]}

\usepackage{dsfont}
\usepackage{tikz}
\usetikzlibrary{positioning,backgrounds,patterns,calc,fit,shapes,matrix}
\tikzstyle{vert}=[circle,draw=black,minimum size=8pt,inner sep=1pt]
\tikzstyle{vertex2}=[circle,draw=black,minimum size=15pt,inner sep=2pt]
\tikzstyle{edge}=[]
\tikzstyle{ypath}=[ultra thick]
\tikzstyle{dottedEdge}=[dotted,thick]
\tikzstyle{small-vertex}=[circle,draw=black,minimum size=6pt,inner sep=0pt,fill=white]
\tikzstyle{thinedges}=[draw=gray!30]
 \tikzstyle{boxes}=[draw,thick, rounded corners=3mm,text width=2.7cm,align=center,text opacity=1,fill opacity=1,fill=white]
\tikzstyle{unk}=[fill=gray!25!white]
\tikzstyle{nodest}=[circle,draw,minimum size=0.55cm,inner sep=.7pt]

\usepackage[margin=0pt,position=t]{subcaption}
\usepackage{mathtools}
\usepackage{cleveref}

\newtheorem{corollary}{Corollary}

\newtheorem{observation}{Observation}

\crefname{table}{Table}{Tables}
\crefname{figure}{Figure}{Figures}
\crefname{cor}{Corollary}{Corollaries}
\crefname{step}{Step}{Steps}
\crefname{rrulev}{Rule}{Rules}
\crefname{thm}{Theorem}{Theorems}
\crefname{obs}{Observation}{Observations}
\crefname{lem}{Lemma}{Lemmas}
\crefname{claim}{Claim}{Claims}
\crefname{section}{Section}{Sections}
\crefname{subsection}{Section}{Sections}
\crefname{figure}{Figure}{Figures}
\crefname{algorithm}{Algorithm}{Algorithms}
\crefname{proposition}{Proposition}{Propositions}
\crefname{theorem}{Theorem}{Theorem}
\crefname{lemma}{Lemma}{Lemmas}
\crefname{construction}{Construction}{Constructions}
\crefalias{AlgoLine}{line}
\crefname{observation}{Observation}{Observations}

{\bfseries}{\normalfont}
\crefname{rrule}{Rule}{Rules}

\newcommand{\decprob}[3]{%
  \begin{center}%
    \begin{minipage}{0.95\linewidth}%
      \textsc{#1}\\
      \textbf{Input:} #2\\
      \textbf{Question:} #3
    \end{minipage}%
  \end{center}%
}

\makeatletter
\newcommand{\gettikzxy}[3]{%
  \tikz@scan@one@point\pgfutil@firstofone#1\relax
  \edef#2{\the\pgf@x}%
  \edef#3{\the\pgf@y}%
}
\makeatother

\pgfdeclarelayer{background}
\pgfdeclarelayer{foreground}

\newcommand{\select}{+}
\newcommand{\unselect}{-}
\newcommand{\force}{*}
\newcommand{\tbf}[1]{\textbf{#1}}

\usepackage{cases}
\usepackage{microtype,ellipsis}
\usepackage{color}

\usepackage{units}

\makeatletter
\def\NAT@spacechar{~}
\makeatother

\usepackage{colortbl}
\usepackage{xcolor}
\usepackage{paralist}
\usepackage{enumerate}
\usepackage[inline]{enumitem}

\usepackage{etoolbox}
\usepackage{appendix}

\newcommand{\EFDGN}{\textsc{EF-DG~(\#)}\xspace}
\newcommand{\EFODGN}{\textsc{EF1-DG~(\#)}\xspace}

\newcommand{\EFDGW}{\textsc{EF-DG~(\euro)}\xspace}
\newcommand{\EFODGW}{\textsc{EF1-DG~(\euro)}\xspace}

\newcommand{\EFDGNW}{\textsc{EF-DG~(\#/\euro)}\xspace}
\newcommand{\EFODGNW}{\textsc{EF1-DG~(\#/\euro)}\xspace}

\newcommand{\EFDG}{\textsc{EF-DG}\xspace}
\newcommand{\EFODG}{\textsc{EF1-DG}\xspace}

\newcommand{\subsetSum}{\textsc{Subset Sum}\xspace}
\newcommand{\partition}{\textsc{Partition}\xspace}
\newcommand{\setCover}{\textsc{Set Cover}\xspace}
\newcommand{\kSum}{\textsc{$k$-Sum}\xspace}
\newcommand{\kpSum}{\textsc{$\ge k$-Sum}\xspace}
\newcommand{\kmSum}{\textsc{$\le k$-Sum}\xspace}
\newcommand{\minsubsetSum}{\textsc{Minimum Size Subset Sum}\xspace}
\newcommand{\maxsubsetSum}{\textsc{Maximum Size Subset Sum}\xspace}
\newcommand{\mcknapsack}{\textsc{Multiple-Choice Knapsack}\xspace}
\newcommand{\km}{k^-}
\newcommand{\kp}{k^+}

\newcommand{\lm}{\ell^-}
\newcommand{\lp}{\ell^+}

\newcommand{\wa}{w_a}
\newcommand{\wrr}{w_r}
\usepackage{eurosym}

\usepackage{thmtools}
\usepackage{thm-restate}

\title{Multivariate Algorithmics for Eliminating Envy  by Donating Goods}

\author[1]{Niclas Boehmer}
\author[2]{Robert Bredereck}
\author[1]{Klaus Heeger}
\author[3]{Du\v{s}an Knop}
\author[4]{Junjie Luo}

\affil[1]{ \small Technische Universit{\"a}t Berlin,  Algorithmics and Computational Complexity, Berlin, Germany\protect\\
\{niclas.boehmer,heeger\}@tu-berlin.de}
\affil[2]{ \small Humboldt-Universit{\"a}t zu Berlin, Institut f{\"u}r Informatik, Algorithm Engineering, Berlin, Germany\protect\\ robert.bredereck@hu-berlin.de}
\affil[3]{ \small Czech Technical University in Prague, Prague, Czech Republic\protect\\ dusan.knop@fit.cvut.cz}
\affil[4]{ \small Nanyang Technological University, Singapore\protect\\ junjie.luo@ntu.edu.sg}

\date{\today}

\begin{document}

\maketitle

\begin{abstract}
	Fairly dividing a set of indivisible resources to a set of agents is of utmost importance in some applications.
	However, after an allocation has been implemented the preferences of agents might change and envy might arise.
	We study the following problem to cope with such situations:
	Given an allocation of indivisible resources to agents with additive utility-based preferences, is it possible to socially donate
	some of the resources (which means removing these resources from the allocation instance)
	such that the resulting modified allocation is envy-free (up to one good).
	We require that the number of deleted resources and/or the caused utilitarian welfare loss of
	the allocation are bounded. We conduct a thorough study of the
	(parameterized) computational complexity of this problem considering various
	natural and problem-specific parameters (e.g., the number of agents, the number
	of deleted resources, or the maximum number of resources assigned to an agent
	in the initial allocation) and different preference models, including
	unary and 0/1-valuations.
	In our studies, we obtain a rich set of (parameterized) tractability and intractability results
	and discover several surprising contrasts, for instance, between the two
	closely related fairness
	concepts envy-freeness and envy-freeness up to one good and between the
	influence of the
	parameters maximum number and welfare of the deleted resources.
\end{abstract} 

\section{Introduction}

Zarah is in big troubles due to numerous complains about an unfair allocation of resources.
Alice thinks that Bob is much better off because of his new screen and laptop.
Bob and Carol explain that Dan's room is way bigger than that of everyone else
and that only his key can unlock the backdoor.
Dan complains about Alice owning a new tablet and keyboard,
while everyone admires Carol's mouse and that she has a large fridge in her room just for herself.
While we leave it to the reader whether the protesters are Zarah's children, PhD students
or employees, we remark that this situation needs to be cleared quickly,
because envy blocks Zarah's prot\'{e}g\'{e}s from doing anything other than complaining.
Since an envy-free reallocation of the resources turned out to be impossible in
her case,
Zarah implements another effective solution:
She decides that Carol's fridge is now usable by everyone and puts it into Dan's room to take away his extra space
while donating Alice's keyboard and Bob's screen to the orphanage.
Doing so, Zarah completely eliminates all envy between her prot\'{e}g\'{e}s.

Real-world allocations, as in our toy example, are often not envy-free for
various reasons (even in envy-free allocations, envy can emerge if
preferences change over time).
While reallocating resources might generally be an option, this can
be very expensive or even impossible (e.g., Alice's room might be too small for Carol's fridge). Moreover, envy-free (re)allocations may simply not exist.
Nevertheless, the need for envy-freeness is undoubtful in some applications
(such as heritage or divorce disputes), so that every possible way
out should be considered.
This work focuses on one of the most natural such possibilities:
Given an allocation of resources, we ask to make the allocation envy-free (EF)
or envy-free up to one good (EF1) by donating, that is, taking away some of the resources.
Since empty allocations are envy-free but obviously unwanted, we additionally
aim for an
upper bound on the number of resources being donated and/or on the social
welfare decrease.

There is also a more subtle interpretation of donating a resource as sharing it, i.e., making it accessible to everyone. If reallocating resources is impossible or not sufficient, then making a resource accessible to everyone (like Carol's fridge in our example) can be a very attractive way out, increasing the overall social welfare.

\subsection{Related Work}
We are not aware of previous work capturing exactly our setting but remark that
the idea of improving the fairness of an allocation is trending and considered
in many different ways and models.
We give an overview of the most related work.

Segal-Halevi~\cite{DBLP:conf/ijcai/Segal-Halevi18} studied a model similar to
ours for divisible resources, where an initial \emph{unfair} allocation of a
cake is given, and the goal is to redivide the cake to balance fairness and
ownership rights.
The two most prominent differences to our work are that he
considered divisible resources (in contrast to indivisible resources) and that
he assumed that resources can only be reallocated and not donated (while we do
not allow for reallocating resources but only for donating them).

For indivisible resources, assuming that agents have linear preferences over
individual resources, Aziz et al.~\cite{DBLP:conf/ijcai/AzizSW16} studied
the problem of adding/deleting a minimum number of resources from a given
set of resources such that an 
envy-free allocation of the remaining resources exists.
Our model differs from this in that we assume
additive utility-based preferences, we modify a given initial allocation, and
we aim for bounding the number of deleted resources and/or
the decrease in social welfare.
In a follow-up work to the work of Aziz et al., Dorn et
al.~\cite{DBLP:journals/algorithmica/DornHS21} considered the goal of modifying
a given set of resources such that there exists a
proportional allocation, also in the
setting with ordinal preferences.
Notably, besides the case without an initial allocation, they also considered
the variant where an initial allocation is given as in our model.

To settle up the existence of envy-free up to any good (EFX) allocations for
indivisible resources, a series of works considered finding partial allocations
(where some resources are allowed to remain unallocated; the unallocated
resources can be interpreted as being donated) that satisfy EFX and
have good
qualitative and/or quantitative guarantees on the
unallocated resources.
Caragiannis et al.~\cite{DBLP:conf/ec/CaragiannisGH19} showed that there always exists a partial EFX allocation whose Nash welfare is at least half of the maximum possible Nash welfare for the full set of resources.
Chaudhury et al.~\cite{DBLP:conf/soda/ChaudhuryKMS20} showed that there always exists a partial EFX allocation such that the number of unallocated resources is bounded by the number of agents minus one.
Our goal is also to find fair partial allocations with good guarantees on the
unallocated resources, but we assume that an initial allocation is already given and we study different fairness notions (EF and EF1).

To overcome that envy-free allocations often do not exist when dividing a set
of indivisible resources and to generally improve
existing allocations,
many more approaches have been considered, 
including making a few resources divisible or sharable~\cite{brams1996procedure,DBLP:journals/corr/abs-1908-01669,VK13,BKK14,BKLNS22},
subsidies and money transfers~\cite{DBLP:conf/aaai/000121,DBLP:journals/corr/abs-2002-02789,DBLP:conf/sagt/HalpernS19,DBLP:journals/corr/abs-2106-00841},
or reallocation of some
resources~\cite{DBLP:conf/atal/AzizBLLM16,DBLP:conf/atal/Endriss13,DBLP:conf/ijcai/GourvesLW17}.

\begin{table}[t]
\caption{Overview of our results for agents with identical binary
valuations. Note that all hardness results hold for only two agents.
Additionally, we prove in \Cref{EFON-n=2} that for non-identical binary
valuations \EFDGN is W[1]-hard with respect to $\km$ and W[1]-hard with
respect to $\kp$
even for only two agents.}
\begin{center}
\resizebox{0.8\textwidth}{!}{\begin{tabular}{l|l|l|l}
    & \textsc{DG~(\#)}       				& \textsc{DG~(\euro)} &
\textsc{DG} \\
\midrule
EF1   &  $\mathcal{O}(|\mathcal{I}|+m\log m)$ (Th. \ref{ib:EF1-N})	&
NP-h. (Pr. \ref{ib:ONP}) & $\mathcal{O}(({\lp})^6+|\mathcal{I}|)$
(Th. \ref{ib:EF1-lp})   \\
   &  	& 
& $\mathcal{O}(({\lm})^6+|\mathcal{I}|)$
(Th. \ref{ib:EF1-lm})
   \\
	& 	& 
& W[1]-h. wrt. $\km$ (Pr. \ref{ib:EF1-kp-km})
	  \\
	& 	&  & NP-h. for $\kp=0$ (Pr. \ref{ib:ONP})
	  \\
\midrule
EF   &  NP-h. for $\kp=1$ (Th. \ref{ib:kmkp})	&
NP-h. for $\lp=1$ (Th. \ref{ib:kmkp}) & $\mathcal{O}(({\lm})^5+|\mathcal{I}|)$
(Th. \ref{ib:EF-general})   \\
   &  W[1]-h. wrt. $\km$ (Th. \ref{ib:kmkp})	&
	  \\
\end{tabular}}
\label{tab:results_bin}
\end{center}
\end{table}

\subsection{Our Contributions}
We study the complexity of \textsc{EF/EF1 By Donating
Goods} (\textsc{EF-/EF1-DG}), where given an initial allocation and two
integers $\km$ and $\lp$ the question is whether we can donate at most~$\km$
resources such that the resulting allocation satisfies EF/EF1 and has
utilitarian welfare at least~$\lp$. Apart from this we also consider two
 special cases of this problem: \textsc{EF-/EF1-DG~(\#)} where we set
$\lp$ to zero and \textsc{EF-/EF1-DG (\euro)} where we set $\km$ to the
number $m$ of~resources.

We split the paper into two parts. First, in \Cref{sec:bp}, we start with the case of binary-encoded valuations. Second, in \Cref{sub:uv} we consider the computationally easier case with unary-encoded valuations (which is quite realistic to assume, as in most applications valuations should be ``small'' numbers). All hardness results for unary valuations directly imply hardness for binary valuations and all algorithms for binary valuations are also applicable to unary valuations. Moreover, notably, some of our algorithmic results from \Cref{sub:uv} initially designed for unary valuations also work for binary valuations.

In \Cref{sec:bp}, where we assume binary valuations, we mostly focus on the natural special case of identical valuations, as otherwise our problems remain computationally intractable even under quite severe restrictions.
We conduct a
complete parameterized analysis taking into account the
following five
parameters: \begin{enumerate*}[label=(\roman*)]
\item the number $n$ of agents,
\item the number $\km$ of donated resources,
\item the number $\kp$ of remaining resources,
\item the summed welfare $\lm$ of donated resources, and
\item the summed welfare $\lp$ of the resulting allocation.
\end{enumerate*}
An
overview of our results from this part
can be found in \Cref{tab:results_bin}:
Considering EF and identical binary valuations, we prove that
our problems are NP-hard even for constant values of $n$, $\lp$, and~$\kp$.
For the dual parameters $\km$ and $\lm$, while \EFDGN is W[1]-hard with
respect to $\km$, the general problem \EFDG is solvable in time polynomial in
$\lm$.
Switching from EF to EF1 leads to further tractability results: \EFODGN
becomes solvable in polynomial-time, while \EFODG is now solvable in time
polynomial in~$\lm$ or~$\lp$.
Our results  reveal several interesting contrasts. First,
they suggest that for identical valuations
computational problems for EF1 are easier to solve than for EF, which
on a high-level is due to the fact that there always exists an EF1 sub-allocation
where the least happy agents in the initial allocation get their full bundle (i.e., they keep their original bundle).
Second,
at least for EF1, \textsc{DG~(\#)} is easier to solve than \textsc{DG~(\euro)}.
Third,
the parameters $\lm$ and $\lp$ seem to be more powerful than the
related parameters $\km$ and $\kp$.

In \Cref{sub:uv}, where  we assume that the agents have unary-encoded arbitrary (non-identical)
valuations, in addition to the parameters considered in the first part,
we also consider \begin{enumerate*}[label=(\roman*)]
\item[(vi)] the maximum number $d$ of resources allocated to an agent in the
initial allocation,
\item[(vii)] the maximum number $\wa$ of resources an agent values as non-zero,
\item[(viii)] the maximum number $\wrr$ of agents that value a resource as
non-zero, and
\item[(ix)] the maximum utility~$u^*$ that an agent assigns to a
resource.\footnote{We do not consider these parameters in the
first part
because the parameter $\wrr$ is upper-bounded by the number of agents.
Moreover, for the parameter $u^*$ it makes no difference whether valuations
are encoded in binary or unary. For the parameters $d$ and $\wa$, our hardness
results and some of our algorithmic result from the second part translate
to the first part.}
\end{enumerate*} An overview of these results grouped by the involved parameters can be found in \Cref{tab:results_1,tab:results_2,tab:results_3}
in \Cref{sub:uv}. We discuss the results in
detail in \Cref{sub:uv}, but provide four highlights already here:
\begin{itemize}
 \item We identify two relevant cases where one of our problems is fixed-parameter
 tractable in one of our parameters alone: \textsc{EF-/EF1-DG (\euro)}
 is fixed-parameter tractable with respect to $d$ or $\wa$
 (\Cref{thm:welfaredwa}). In contrast to this,
 \textsc{EF-/EF1-DG (\#)} is NP-hard for constant $d+\wa+\wrr$ for
 0/1-valuations (\Cref{nkmw}). Notably, in all other cases, the
 complexity of \textsc{EF-DG (\euro)} and \textsc{EF-DG (\#)} and of \textsc{EF1-DG (\euro)} and \textsc{EF1-DG~(\#)} are the same.
 \item Even for unary-encoded valuations, our problems are mostly
computationally intractable. With the above mentioned exception, none of our
parameters alone is sufficient to
establish fixed-parameter tractability. Nevertheless, we still identify
some tractable and realistic cases, for instance, when the number of agents and
the maximum utility value is bounded (\Cref{nu}) or when the number/welfare of the deleted
resources and the maximum number of resources assigned to an agent in the initial
allocation is bounded~(\Cref{vkm}).
 \item In contrast to the binary case, the number $\km$ of
resources to delete is a more powerful parameter than the welfare $\lm$ of the
deleted resources, as for some parameter combinations involving $\km$ our
general problem is fixed-parameter tractable but not for the respective
combination involving~$\lm$.
\item 
The
number~$\kp$/welfare~$\lp$ of the remaining resources is a less powerful parameter than the
number~$\km$/welfare~$\lm$ of the deleted resources, because $\kp$/$\lp$ constitutes a lower bound (on the remaining resources/welfare), while $\km$/$\lm$ is an upper bound (on the deleted resources/welfare).
\end{itemize}

\section{Preliminaries}

\paragraph{Resource Allocation.}
We consider a set $A=\{a_1,\dots, a_n\}$ of $n$ \emph{agents} and a set
$R=\{r_1,\dots,r_m\}$ of $m$ \emph{resources}.  We assume that the agents have
additive
and cardinal preferences:
For each agent~$a\in A$ and resource $r\in R$, let $u_a(r)\in \mathbb{N}_{0}$
denote the \emph{utility}
agent $a$ assigns to resource~$r$. In this context, we also say that $a$ values
$r$ as $u_a(r)$. We denote as $u^*$ the maximum utility value
an agent assigns to a resource, i.e., $u^*:=\max_{a\in A, r\in R} u_a(r)$.
Further, for a subset~$R'\subseteq R$ of
resources, we set
$u_a(R')=\sum_{r\in R'} u_a(r)$. We say that agents have \emph{binary/unary
valuations} if for all $a\in A$ and $r\in R$ the utility values~$u_a(r)$
in the input are encoded in binary/unary. Further, we say that agents have
\emph{0/1-valuations} if $u_a(r)\in \{0,1\}$ for all $a\in A$ and $r\in R$.

In
our parameterized analysis, we denote as $\wa\in [m]$ the maximum number of
resources
an agent values as non-zero, i.e., $\wa:=\max_{a\in A} |\{r\in R\mid
u_a(r)>0\}|$
and as $\wrr\in [n]$ the maximum number of agents that value a resource as
non-zero,
i.e., $\wrr:=\max_{r\in R} |\{a\in A\mid u_a(r)>0\}|$.

An \emph{allocation $\pi$ of resources to agents} is a function $\pi\colon
A\mapsto
2^R$ such that $\pi(a)$ and $\pi(a')$ are disjoint for $a\neq a'\in A$.
For an
agent~$a\in A$ and an allocation $\pi$, $\pi(a)\subseteq R$ is
the set of resources assigned to agent $a$ in $\pi$. We also refer
to $\pi(a)$ as $a$'s \emph{bundle} in $\pi$. For
readability, we always assume that initial allocations are complete,
i.e., for every $r\in R$ there exists an agent $a\in A$ such that $r\in \pi(a)$.  We
refer to $\sum_{a\in
A}u_a(\pi(a))$ as
the \emph{(utilitarian) welfare} of~$\pi$.

For two agents $a\neq a'\in A$ and an allocation $\pi$, we say that agent~$a$
\emph{envies}
agent~$a'$ in $\pi$ if $a$ prefers $a'$'s bundle to its own bundle, i.e.,
$u_a(\pi(a'))>u_{a}(\pi(a))$. Further, we say that agent~$a$ \emph{envies agent~$a'$ up to one good} in $\pi$ if, for all resources $r\in \pi(a')$, agent~$a$
prefers $\pi(a')\setminus \{r\}$ to~$\pi(a)$, i.e., $\min_{r\in
\pi(a')}u_{a}(\pi'(a)\setminus
\{r\})>u_a(\pi(a))$.
An allocation is \emph{envy-free} (EF) [\emph{envy-free up to one
good}
(EF1)] if there is
no agent that envies another agent [up to one good].

\paragraph{Envy-Freeness by Donating Goods.}
We now define our central problem. In the following, for clarity, as donating a
resource corresponds to deleting the resource from the instance,
we say that a resource is deleted if it is donated.

 \decprob{EF/[EF1] by Donating Goods (\EFDG/[\EFODG{}])}
{Given a set $A$ of agents, a set $R$ of resources,
	an initial
	allocation $\pi$ of resources to agents, and integers $\km$ and $\lp$.}
{Is it possible to delete at most $\km$ resources from~$\pi$ such that the
	resulting
	allocation $\pi'$ is envy-free [up to one good] and has utilitarian welfare
	at least $\lp$,
	i.e., $\pi'(a)\subseteq \pi(a)$ for all $a\in A$, $\sum_{a\in A}
	|\pi(a)\setminus \pi'(a)|\leq \km$, and $\sum_{a\in A} u_a(\pi'(a))\geq
	\lp$?
}
We also consider two special cases of this general problem. That
is, the so-called number variant (\EFDGN/\EFODGN) where we only have a
bound on the number of deleted resources, i.e., $\lp=0$, and the so-called
welfare variant (\EFDGW/\EFODGW) where we only have a bound on the
utilitarian welfare of the resulting allocation, i.e., $\km=m$. For notional
convenience, we write \EFDGNW and \EFODGNW to refer to both the
number and welfare~variant.

Apart from $\km$ and $\lp$, we also consider the respective dual
parameters.
That is, the minimum number~$\kp:=m-\km$ of remaining resources and the maximum welfare
${\lm:=\sum_{a\in A} u_a(\pi(a))-\lp}$ of the deleted resources.
In addition, we consider the maximum number $d:=\max_{a\in A} |\pi(a)|$ of
resources that an agent holds in the initial allocation $\pi$ as a
parameter.
Throughout the paper, we assume that basic arithmetic operations (i.e.,\
addition
and subtraction) of natural numbers can be
performed in constant time.

\paragraph{Auxiliary Problems.} \label{sub:prob}

We introduce problems from which or to which we
reduce.
Given positive integers $x_1, x_2, \ldots, x_\nu$,
and $S$, \subsetSum is the problem to decide whether there is a set $I
\subseteq [\nu]$ such that $\sum_{i \in I} x_i = S$.
\partition is the \subsetSum problem with $S = \frac{1}{2}\sum_{i
\in [\nu]} x_i$.
Both \subsetSum and \partition are NP-hard assuming that numbers are encoded in
binary~\cite{DBLP:conf/coco/Karp72}.
The \minsubsetSum problem (resp. the \maxsubsetSum problem) is to find the size
of the minimum (resp. maximum) subset whose sum is equal to $S$.
\subsetSum and \textsc{Minimum/Maximum Subset Sum} can be solved in
$\mathcal{O}(\nu S)$ time by
dynamic programming~\cite{https://doi.org/10.1002/nav.3800030107}.
\subsetSum with the goal of finding a subset of size
exactly $k$ (i.e., $|I|=k$) is called \kSum. It is W[1]-hard with respect to $k$~\cite{DBLP:journals/tcs/DowneyF95}.
The \kmSum problem (resp. the \kpSum problem) is the \subsetSum problem with a
upper bound (resp. lower bound) for the size of the solution, i.e., $|I|\le k$
(resp.~$|I| \ge k$) and is W[1]-hard with respect to $k$. To see this, there is
an easy reduction from \kSum to \kmSum (and \kpSum) by replacing $x_i$
with $x_i+T$ and replacing $S$ with $S+kT$, where $T$ is a large number (e.g.,
$T=\sum_{i \in [\nu]} x_i$).

We reduce some of our problems to:
\decprob{\mcknapsack}
{A capacity $c$, an integer $k$, and $\ell$ sets $S_i$ $(i \in [\ell])$, each
resource $j \in S_i$ has a profit~$p_{i,j}$ and a weight $w_{i,j}$.}
{Is it possible to choose for each $i\in [\ell]$ one resource~$j_i$ from $S_i$
such that
$\sum_{i \in [\ell]} w_{i,j_i} \le c$ and $\sum_{i \in [\ell]} p_{i,j_i} \ge
k$?}

The \mcknapsack problem can be solved in $\mathcal{O}(c \cdot \sum_{i \in [\ell]}|S_i|)$ time by dynamic programming \cite{dudzinski1987exact}.

Finally, we define the \setCover problem.
\decprob{\setCover}
{A universe $S$, a family of sets $\mathcal{C} \subseteq 2^S$, and a positive
integer~$z$.}
{Is there a collection $\mathcal{C'} \subseteq \mathcal{C}$ of size at most $z$
such that $S = \bigcup_{C \in \mathcal{C'}} C$?}
\textsc{Restricted Exact Cover by 3-Sets} is a NP-hard
\cite[Problem SP2]{DBLP:books/fm/GareyJ79} special case of
\textsc{Set Cover} where each set consists of exactly three elements and each
element appears in exactly three sets.

\section{Binary Valuations} \label{sec:bp}
In this section, we assume that the valuations of agents are encoded in binary
and mostly (except from \Cref{sub:EF1UD}) focus on situations where agents have
identical valuations (as we obtain NP-hardness even in this case).
We start by considering EF1 (\Cref{sub:EF1ID,sub:EF1UD}) and afterwards turn
to EF (\Cref{sub:EFID}).

For identical valuations, we assume that there are no resources
that
are valued as $0$ by all agents, as
otherwise we can remove them and update $\kp$ and $m$ to get an
equivalent instance.
For identical valuations, we will
call the utility function of all agents~$u$.

\subsection{EF1 and Identical Valuations}\label{sub:EF1ID}
In this section, we analyze \EFODG and its two special cases \EFODGNW assuming
that agents have identical binary valuations and
identify several tractable cases for
these problems. We start with two
general results concerning EF1 allocations. First, we observe that as
valuations are identical, an agent $a$ is envied by another agent $a$ up to one good
if it is envied by a least happy agent up to one good:
\begin{observation}
	\label{ib:EF1-chara}
	For identical valuations, an allocation is EF1 if and only if the agents
	with the minimum utility in the allocation do not envy other agents up to
	one good.
\end{observation}

Moreover, for identical valuations, if we are
given an EF1 allocation $\pi'$ with $\pi'(a)\subseteq \pi(a)$ for each $a\in A$, then from this we can
construct an EF1
allocation where all agents $a\in A$ get at least~$\pi'(a)$ and all
agents $a\in A$ that have the minimum utility in $\pi$ get their full bundle~$\pi(a)$ in $\pi'$. This can be done by successively adding for an agent $a\in A$ who has the minimum utility under~$\pi'$ and fulfills~$\pi(a) \neq \pi'(a)$ an arbitrary resource from~$\pi(a) \setminus \pi'(a)$.
\begin{restatable}{lemma}{EFOL}\label{ib:EF1-least}
	Let $A_0$ be the set of agents that have the minimum utility in the
	initial allocation.
	For identical valuations, if there exists a solution for an instance
	\EFODG (or \EFODGN or \EFODGW), then there is a solution such that all
	agents from $A_0$ get their full initial bundle
	and all other agents have higher utility than them.
\end{restatable}
\begin{proof}
	Given a solution $\pi'$ for an instance of \EFODG (or \EFODGNW), we
construct a solution
	satisfying the stated constraints.
	As long as there is an agent $a\in A$ who has the minimum utility under
	$\pi'$
	and fulfills $\pi(a) \neq \pi'(a)$, we can add an arbitrary resource from~$\pi(a) \setminus \pi'(a)$ to~$\pi'(a)$ and update $\pi'$ accordingly.
	During this process, we always maintain that $\pi'$ satisfies EF1 and
that $\pi'$ is a solution to the given \EFODG instance.
	When the process stops, as the agent with minimum utility needs to have
its full bundle, all agents in $A_0$ must have the minimum utility
	in $\pi'$ and they must get their full initial bundle (i.e.,
	$\pi(a)=\pi'(a)$)
	while all agents in $A \setminus A_0$ need to have a higher utility than
	agents in
$A_0$.
\end{proof}

Using \cref{ib:EF1-chara,ib:EF1-least}, one can reduce \EFODGN to
finding for each agent~$a\in A\setminus A_0$ a minimum size set
$P_a\subseteq \pi(a)$ such that deleting it makes agents from $A_0$ no longer envy
agent $a$ up to one good. This
problem can be solved using a simple greedy algorithm by moving the most
valuable resources from $\pi(a)$ to~$P_a$ until all envy is resolved. This
establishes that
\EFODGN is
polynomial-time solvable for identical binary valuations:

\begin{restatable}{theorem}{EFON}\label{ib:EF1-N}
	For identical valuations, \EFODGN can be solved in
$\mathcal{O}(|\mathcal{I}|+m\log m)$ time.
\end{restatable}
\begin{proof}
	Let $A_0$ be the set of agents who have the minimum utility under the
initial allocation
	and $A^* \subseteq A \setminus A_0$ be the set of agents that are envied by
agents from $A_0$ up to one good.
	By \cref{ib:EF1-chara,ib:EF1-least}, to solve the problem, for each $a \in
	A^*$, it suffices to
find a subset of resources $P_a \subseteq \pi(a)$ of minimum size such
that
after deleting all resources from $P_a$, agent $a$ will not be envied by agents
from
$A_0$ up to one good.
	We can solve this for each agent $a \in A^*$ by greedily deleting resources
	from~$\pi(a)$ with the highest
value till agent $a$ is not envied by agents in $A_0$ up to one good.
	Let $P_a$ be the set of the deleted resources in the bundle of $a\in A^*$.
	Delete resources in $P_a$ for each agent $a \in A^*$
then results in an EF1 allocation:
No agent from~$A_0$ envies another agent up to one good because otherwise we
would have deleted more resources from this agent. Note that each agent $a\in
A^*$ has at least the utility of agents in $A_0$, as otherwise agents from
$A_0$ would not envy $a$ up to one good before the deletion of the last
resource from $\pi(a)$.
Thus, as agents in~$A_0$ do not envy any
agent up to one good,
no agent~$a \in A^*$ envies another agent up to one good.
	Let $s_a=|P_a|$. We claim that for any EF1 allocation, we have to delete
at least $s_a$ resources from agent $a$'s bundle.

	Indeed, if deleting $s'_a<s_a$ resources from agent $a$'s bundle is enough
	to
	get an EF1 allocation, then by the definition of EF1, we can delete $s'_a+1
\le s_a$
	 resources from $\pi(a)$ such that agent~$a$ is not envied by agents in
	 $A_0$,
	 which is a contradiction to the choice of $P_a$ (as this implies that
agents from $A_0$ already would not have envied $a$ up to one good before the
deletion of the last resource from $\pi(a)$).
	 Finally we check whether $\sum_{a \in A^*} s_a \le \km$.
	 The running time is $\mathcal{O}(|\mathcal{I}|+m\log m)$ as we need to
sort the
resources
	 for each $a \in A^*$ by their values.
\end{proof}

In contrast to this, for \EFODGW, solving said subproblem of
finding for each agent who is not part of $A_0$ a subset of its
initial bundle with minimum summed welfare such that deleting it makes agents
from $A_0$ no longer envy the
agent up to one good basically requires solving an instance of
\textsc{Subset Sum}. Indeed, by reducing from
\textsc{Partition}, we can show
that in contrast to \EFODGN, \EFODGW is NP-hard for two agents with identical
binary valuations:

\begin{restatable}{proposition}{ibONP}\label{ib:ONP}
	For identical binary valuations, \EFODGW with $n=2$ agents is NP-hard.
\end{restatable}
\begin{proof}
	We present a reduction from \partition.
	Let $x_1, x_2, \ldots, x_\nu$, and $S$ be an instance of \partition with
	$2S=\sum_{i \in [1,\nu]} x_i$.
	We create an instance with two agents and $\lm \coloneqq S$.
	For each integer $x_i$, agent one holds one resource of value $x_i$.
	Additionally, agent one holds an resource of value $S$.
	Agent two holds one resource of value $S$.
	Then if the instance of \partition is a Yes-instance, we can delete the
	corresponding resources from agent one's bundle with utility $S$ to create
	an EF1 allocation where agent
	one has value $2S$ and agent two has value $S$.
	If the instance \partition is a No-instance, then, since in any EF1 allocation
	the utility of agent one should be reduced by at least $S$, this implies that
	we have to delete resources from agents one's
	bundle with utility more than $S$ to get an EF1 allocation.
\end{proof}

Our hardness reduction from \Cref{ib:ONP} does not have any
implications on the parameterized complexity of \EFODGW with respect to
$\lp$ respectively $\lm$. In fact, parameterized by~$\lp$ respectively $\lm$, \EFODGW (and even the
general problem \EFODG) become tractable. This stands in contrast with
the previously proven NP-hardness for \EFODGW, which implies that \EFODG is NP-hard for $\kp=0$.
While this contrast between $\kp$ and $\lp$ might look
surprising at first glance, recall that $\lp$ bounds the otherwise binary
encoded welfare of the deleted resources. Thus, it is quite intuitive that
$\lp$ is more powerful than~$\kp$.

\begin{restatable}{theorem}{EFOlm}\label{ib:EF1-lp}\label{ib:EF1-lm}
	For identical valuations, \EFODG can be solved in
$\mathcal{O}((\lp)^6+|\mathcal{I}|)$ or $\mathcal{O}((\lm)^6+|\mathcal{I}|)$
time.
\end{restatable}
\begin{proof}
\textbf{Parameter $\ell^+$.}
	We begin with some pre-processing.
	Let $u_0=\min_{a \in A} u(\pi(a))$.
	If $u_0=0$, then in an EF1 allocation $\pi'$, each agent can hold at most
	one resource with non-zero value. Thus, in an optimal solution, each agent gets assigned its
	most
	valuable resource. In the following, we assume $u_0\geq 1$.
	Using \Cref{ib:EF1-N}, we can get an EF1 allocation $\pi'$ with the
	maximum number of resources left but without any guarantee on the welfare of
	$\pi'$.
	If $u_0 \ge \lp$ or $n \ge \lp$, then $\sum_{a\in A} u(\pi'(a)) \ge n u_0
	\ge \lp$ is guaranteed.
	In the following we assume $0<u_0 < \lp$ and~$n < \lp$.
	Let $r^*$ be the resource with maximum value and $a^*$ be the agent such that $r^* \in \pi(a^*)$.
	If $u(r^*) \ge \lp$, then either $r^* \in \pi'(a^*)$ and hence $\sum_{a\in
	A} u(\pi'(a)) \ge \lp$, or $r^* \not \in \pi'(a^*)$, in which case we can
	replace the most valuable resource from $\pi'(a^*)$ with $r^*$. In both cases, we have
	constructed a solution.
	In the following, we assume that
	$u(r^*) < \lp$.

	Now, we turn to the main part of the algorithm. Let $A_0$ be the set of agents who have the minimum utility in the initial
	allocation
	and $A^* \subseteq A \setminus A_0$ be the set of agents that are envied by
agents in $A_0$ up to one good.
	According to \cref{ib:EF1-chara,ib:EF1-least}, for each $a_i \in A^*$, we
	need to find a set $R_i \subseteq \pi(a_i)$ such that by
	keeping
	all resources from $R_i$ and deleting all resources from $\pi(a_i) \setminus R_i$,
	agent $a_i$ will not
	be envied
	by agents in $A_0$ up to one good, and $\sum_{a_i \in A^*} |R_i| \ge \tilde{k}^+$ and
	$\sum_{a_i \in A^*} u(R_i) \ge \tilde{\ell}^+$, where $\tilde{k}^+ :=\kp - \sum_{a_i \in A \setminus A^*} |\pi(a_i)|$ and $\tilde{\ell}^+ := \lp - \sum_{a_i \in A \setminus A^*} u(\pi(a_i))$.
	We will solve this problem in two steps.
	In Step 1, for each agent $a_i \in A^*$ we compute the set of all possible
	$R_i$ such that after deleting $\pi(a_i)\setminus R_i$ no agent envies
	$a_i$ up to one good.
	Then, in Step 2 we check whether it is possible to select one candidate
	$R_i$ for each $a_i \in A^*$ such that $\sum_{a_i \in A^*} |R_i| \ge \tilde{k}^+$ and
	$\sum_{a_i \in A^*} u(R_i) \ge \tilde{\ell}^+$.

	\textbf{Step 1:}
	Fix some agent $a_i \in A^*$. We want to guess the utility $t$ of~$u (R_i)$
	and the utility
	$t_1$ of the most valuable resource in $R_i$.
	Let $t_2:=t-t_1$ be the utility of the remaining resources in $R_i$.
	By EF1, we have $t_2 \le u_0< \lp$.
	Since $t_1 \le u(r^*)<\lp$, we have $t <2\lp$.
	Thus, we can iterate over at most $2(\lp)^2$ different pairs $(t,t_1)$.
	For each $(t,t_1)$ such that $t_2=t-t_1 \le u_0<\lp$, we find an arbitrary
	resource $r_0 \in \pi(a_i)$ with $u(r_0)=t_1$ (if there is no such
	resource, then
	we can skip this pair), and then, we want to compute the maximum size of a
	subset $R_i=R_i' \cup \{r_0\}$ such that $R'_i \subseteq \pi(a_i) \setminus
	\{r_0\}$ and $u(R'_i) =t_2$. This is an instance of the \maxsubsetSum
	problem (see \Cref{sub:prob}).
	Since $u(R'_i) =t_2 < \lp$, for each value~$v \in [1,\lp]$, set~$R'_i$ can
	contain at most $\lp$ resources with value $v$ (recall that we do not have resources valued as $0$).
	Thus, we can pick a subset $S_i \subseteq \pi(a_i) \setminus \{r_0\}$ of
	resources with $|S_i| \le (\lp)^2$ such that we only need to include
	resources from $S_i$ in $R'_i$.
	We then construct an instance of \maxsubsetSum with target value
	$t_2<\lp$ and set $\{u(r) \mid r \in S_i\}$
	with $|S_i| \le (\lp)^2$. Using dynamic programming, the instance can be
	solved in $\mathcal{O}((\lp)^3)$ time.
	We need to do this for each agent $a_i \in A^*$ and each pair~$(t,t_1)$,
	separately. As we have $n<\lp$ agents and, as observed above, $2(\lp)^2$
	valid pairs, this can be done in $\mathcal{O}((\lp)^6)$ time.
	For each agent $a_i\in A^*$, we get a family $L_i$ of candidate sets (each such set consists of the resources selected by the dynamic program and $r_0$).
	For all candidate sets in $L_i$ with the same utility, we just keep one
	candidate set
	of maximum cardinality.
	Thus, for each $t \in [1,2\lp]$, there is at most one candidate in $L_i$ and
	hence~$|L_i| \le 2\lp$.

	\textbf{Step 2:}
	For each agent $a_i \in A^*$ and each $R_i^j \in L_i$, we compute a pair
$(s_i^j,t_i^j)$,	where $s_i^j=|R_i^j|\le 2\lp$ and $t_i^j=u(R_i^j)\le 2\lp$ resulting in a
set  $Q_i=\{(s_i^j,t_i^j) \mid j\in \{1, \dots, |L_i|\}$ of such pairs.
	Without loss of generality, assume $t_i^1 \ge t_i^j$ for $j \in \{1, \dots, |L_i|\}$.
	Now the problem is to find a pair $(s_i^{j_i},t_i^{j_i})$ from each
$Q_i$ such that $\sum_{a_i \in A^*} s_i^{j_i} \ge \tilde{k}^+$ and $\sum_{a_i \in A^*}
t_i^{j_i} \ge \tilde{\ell}^+$.
	This can be reduced to the \textsc{Multiple-Choice Knapsack Problem} (MCKP)
	(see \Cref{sub:prob})
as follows.
	Let $S^*=\sum_{a_i \in A^*} s_i^1$ and $T^*=\sum_{a_i \in A^*} t_i^1$.
	If $T^* < \tilde{\ell}^+$, then our instance is clearly a No-instance.
	If $T^* \ge \tilde{\ell}^+$ and $S^* \ge \tilde{k}^+$, then we get a solution by selecting the
corresponding set $R_i^1$ for each agent~$a_i \in A^*$.
	Finally, if $T^* \ge \tilde{\ell}^+$ but $S^* < \tilde{k}^+$, then we need to select a pair
$(s_i^{j_i},t_i^{j_i}) \in Q_i$ for each agent $a_i \in A^*$ such that
	$\sum_{a_i \in A^*} (t_i^1-t_i^{j_i}) \le T^*-\tilde{\ell}^+ $ and $\sum_{a_i \in A^*} (s_i^{j_i}-s_i^1) \ge \tilde{k}^+ -S^*$.
	This is an instance of MCKP with sets $L_i$ $(\forall a_i \in A^*)$,
	capacity $T^*-\tilde{\ell}^+$, and lower bound of the target value
	$\tilde{k}^+
	-S^*$. For each $ a_i \in A^*$, each set~$R_i^j \in L_i$ has weight
	$t_i^1-t_i^j$ and value
	$s_i^j-s_i^1$.
	Since $\sum_{a_i \in A^*} |L_i| \le 2(\lp)^2$ and
	$T^*-\tilde{\ell}^+  \le T^* \le |A^*|\cdot \max_{a_i \in A^*} t_i^1 \le 2(\lp)^2$
  the	instance can be solved in $\mathcal{O}((\lp)^4)$ time \cite{dudzinski1987exact}.
	Summing up, our problem can be solved in
	$\mathcal{O}((\lp)^6+|\mathcal{I}|)$ time.

    \bigskip
    \noindent\textbf{Parameter $\ell^-$.}
	Let $A_0$ be the set of agents whose initial bundles have the minimum
	utility~$u_0=\min_{a \in A} u(\pi(a))$
	and $A^* \subseteq A \setminus A_0$ be the set of agents that are envied by
	agents in $A_0$ up to one good.
	According to \cref{ib:EF1-chara,ib:EF1-least}, for each $a_i \in A^*$, we
	need to
	find a subset of resources $P_i \subseteq \pi(a_i)$ such that after deleting
	all
	resources from $P_i$, agent $a_i$ will not be envied by agents in $A_0$ up to
	one
	good, and $\sum_{a_i \in A^*} |P_i| \le \km$ and $\sum_{a_i \in A^*} u(P_i) \le
	\lm$, which implies that $|P_i| \le \km$ and $u(P_i) \le \lm$ for each $a_i
	\in
	A^*$.
	If $|A^*| > \lm$, then we have to delete more than $\lm$ resources and the
	welfare
	will be decreased by more than $\lm$, so the instance is a No-instance.
	In the following, we assume $|A^*| \le \lm$.
	We then solve the problem in two steps, similarly to the case for parameter~$\lp$.
	\medskip

	\textbf{Step 1:}
	For each agent $a_i \in A^*$ and each $t \in [1,\lm]$, we compute a subset
	$P_i \subseteq \pi(a_i)$ with the minimum $|P_i|$ such that $u(P_i)=t$ and
	EF1 is
	guaranteed. Each of these tasks can be reduced to $\mathcal{O}(\lm)$
	instances of the \minsubsetSum problem as follows.
	To guarantee EF1, we first guess the value~$v_i$ of the most valuable resource in
	$\pi(a_i) \setminus P_i$.
	Let $r_i^*$ be the most valuable resource in agent $i$'s initial bundle.
	If $u(r_i^*) \le \lm$, then we guess the value $v_i$ of the most valuable
	resource in
	$\pi(a_i) \setminus P_i$.
	There are at most $\lm$ different choices for $v_i$ and for each choice let
	$r_i$ be an arbitrary resource with $u(r_i)=v_i$.
	If there is no such resource or $u(\pi(a_i))-t-v_i > u_0$, which implies that
	agents from $A_0$ still envy agent $a_i$ if resources $P_i$
	with $u(P_i))=t$ get deleted, then we can skip this choice of $v_i$.
	If $u(r_i^*) > \lm$, since we cannot afford to delete $r_i^*$, we can
	directly conclude that $r_i^*\in \pi(a_i)\setminus P_i$ and thus can just set
	$v_i=u(r_i^*)$ and $r_i=r_i^*$.
	Next we compute the minimum size of a subset
	$P_i \subseteq \pi(a_i) \setminus \{r_i\}$ such that $u(P_i) =t$, which is an
instance of the
	\minsubsetSum problem (as defined in \Cref{sub:prob}).
	Since $u(P_i) =t$, set~$P_i$ can only contain resources from $\pi(a_i)$ with
	value at
	most $t$ and for all resources with the same value, at most $t$ of them
	will be
	contained in $P_i$.
	Thus we can construct a subset $S_i \subseteq \pi(a_i)$ of resources with
	$|S_i| \le t^2
	\le (\lm)^2$ such that we only need to include resources from $S_i$ in $P_i$.
  We then construct an instance of \minsubsetSum with target value
  $t \le \lm$ and set $\{u(r) \mid r \in S_i\}$
  with $|S_i| \le (\lm)^2$. Using dynamic programming, the instance can be
  solved in $\mathcal{O}((\lm)^3)$ time.
  We need to do this for each agent $a_i \in A^*$, each value $t \in [1,\lm]$ for
  the utility of deleted resources,
  and each value $v_i \in [1,\lm] \cup \{u(r_i^*)\}$ for the maximum utility
  among remaining resources separately. As we have $|A^*| \le \lm$, this can be done in $\mathcal{O}((\lm)^6)$ time.
	Then, for each agent $a_i \in A^*$ we get a family $L_i$ of (possibly
	$({\lm})^2$) candidates.
	For all candidates in $L_i$ with the same utility, we just keep one
	candidate that has the minimum size and drop all other candidates.
	Thus for each $t \in [1,\lm]$ there is at most one candidate in $L_i$ and
	hence
	$|L_i| \le \lm$.\medskip

	\textbf{Step 2:}
	For each agent $a_i \in A^*$ and each $P_i^j \in L_i$, we compute a pair
	$(s_i^j,t_i^j)$,	where $s_i^j=|P_i^j|\le \lm$ and $t_i^j=u(P_i^j)\le \lm$ resulting in a
  set $Q_i=\{(s_i^j,t_i^j) \mid j\in \{1, \dots, |P_i^j|\}$ of such pairs.
  Without loss of generality, assume $t_i^1 \le t_i^j$ for $j \in \{1, \dots, |L_i|\}$.
	Now the problem is to find a pair $(s_i^{j_i},t_i^{j_i})$ from each
	$Q_i$ such that $\sum_{a_i \in A^*} s_i^{j_i} \le \km$ and $\sum_{a_i \in A^*}
	t_i^{j_i} \le \lm$.
	This can be reduced to the \mcknapsack problem (MCKP) (as defined in \Cref{sub:prob})
	as follows.
	Let $S^*=\sum_{a_i \in A^*} s_i^1$ and $T^*=\sum_{a_i \in A^*} t_i^1$.
	If $T^* > \lm$, then our instance is clearly a No-instance.
	If $T^* \le \lm$ and $S^* \le \km$, then we get a solution by selecting the
	corresponding set $P_i^1$ for each agent $a_i \in A^*$.
	Finally, if $T^* \le \lm$ but $S^* > \km$, then we need to select a pair~$(s_i^{j_i},t_i^{j_i}) \in Q_i$ for each agent $a_i \in A^*$ such that
	\[ \sum_{a_i \in A^*} (t_i^{j_i}-t_i^1) \le \lm-T^* \quad \text{and} \quad
	\sum_{a_i \in A^*} (s_i^1-s_i^{j_i}) \ge S^*- \km.\]
	This is an instance of MCKP with sets $L_i (\forall a_i \in A^*)$, capacity $\lm-T^*$, lower bound of the target value $ S^*- \km$. For each $ a_i \in A^*$, each set $P_i^j \in L_i$ weight $t_i^j-t_i^1$ and value $s_i^1-s_i^j$.
	Since $\sum_{a_i \in A^*} |L_i| \le (\lm)^2$ and $\lm-T^* \le \lm$, the
	instance can be solved in $\mathcal{O}((\lm)^3)$ time \cite{dudzinski1987exact}.

	Finally, we consider the whole running time.
	Computing the set $A_0$ and $A^*$ can be done in~$\mathcal{O}(|\mathcal{I}|)$~time.
	Computing $L_i$ for all agents in $ A^*$ can be done in
	$\mathcal{O}((\lm)^6)$ time.
	Together with the $\mathcal{O}((\lm)^3)$ time for MCKP, the whole problem
	can be solved in $\mathcal{O}((\lm)^6+|\mathcal{I}|)$ time.
\end{proof}

While we have seen in \Cref{ib:EF1-N} that \EFODGN is polynomial-time solvable,
we now show that \EFODG parameterized by $\km$ is
W[1]-hard even for only two agents, which stands in contrast to the
preceding tractability results for $\lm$:
\begin{restatable}{proposition}{EFOkpkm}\label{ib:EF1-kp-km}
	For identical binary valuations, \EFODG with $n=2$ agents is W[1]-hard
	parameterized by $\km$.
\end{restatable}
\begin{proof}
	We present a reduction from \kmSum, where given positive integers $x_1, x_2, \ldots, x_\nu$, $k$, and $S$, the task is to decide whether there is a set $I \subseteq [\nu]$ with $|I| \le k$ such that $\sum_{i \in I} x_i =S$.
  \kmSum is W[1]-hard with respect to $k$. To see this, there is
an easy reduction from \kSum to \kmSum by replacing $x_i$
with $x_i+T$ and replacing $S$ with $S+kT$, where $T$ is a large number (e.g.,
$T=\sum_{i \in [\nu]} x_i$).
	Let $x_1, x_2, \ldots, x_\nu$, $k$, and $S$ be an instance of \kmSum.
	We create an instance of EF1-BDG with two agents~$ a$ and $b$, $\km:=k$, and $\lm:=S$.
	For each integer $x_i$, agent~$a$ holds one resource of value $x_i$.
	Additionally, agent~$a$ holds an resource $r_1^*$ of value $T$, where
	$T=\sum_{i \in [\nu]} x_i$.
	Agent~$b$ holds one resource $r_2^*$ of value $T-S$.
	Then initially agent one has value $2T$ and in an EF1 allocation, agent~$a$
	can have at most $2T-S$ utility.
	Thus we have to delete a set $P$ of	at most $\km$ resources from~$a$'s bundle such that $u(P) \ge S$.
	Since $\lm=S$, we have $u(P) = S$ and $P \subseteq \pi(1) \setminus
	\{r_1^*\}$.
	Therefore, the instance of \EFODG is a Yes-instance if and only if the
	instance of \kmSum is a Yes-instance.
\end{proof}

\subsection{EF1 and Non-Identical Preferences} \label{sub:EF1UD}
Our positive result for \EFODGN from \Cref{ib:EF1-N} for identical binary
valuations raises
the question whether there is hope for tractability results for \EFODGN
parameterized by $n$, $\km$, or~$\kp$ with
non-identical binary valuations. We answer this question negatively with a
strong
(parameterized) hardness result in the following proposition:

\begin{restatable}{proposition}{EFONn}\label{EFON-n=2}
	For binary valuations, \EFODGN with $n=2$ agents is W[1]-hard parameterized
by $\kp$ or $\km$, even if the two agents only disagree on the valuation of
two resources.
\end{restatable}
\begin{proof}
	We first show the result for $\km$.
	We present a reduction from \kSum.
	The reduction is similar to the one for \cref{ib:EF1-kp-km}.
	Let $x_1, x_2, \ldots, x_\nu$, $k$, and $S$ be an instance of \kSum.
	Let~$T:=\sum_{i \in [\nu]} x_i$.
	We create an instance of \EFODGN with two agents~$a$ and~$b$
	and $\km=k$.
	For each integer~$x_i$, agent~$a$ holds one resource $r_i$ with
	$u_a(r_i)=u_b(r_i)=T+x_i$.
	Additionally, agent~$a$ holds a resource $r_1^*$ with
	$u_a(r_1^*)=u_b(r_1^*)=\nu T$.
	Agent~$b$ holds two resources $r_2^*$ and $r_3^*$ with
	\[u_a(r_2^*)=u_a(r_3^*)=\nu T+(\nu-\km)T+(T-S)\] and
	\[u_b(r_2^*)=u_b(r_3^*)=\frac{(\nu-\km)T+(T-S)}{2}.\]
	Notice that both agents have the same utility for agent~$a$'s bundle, that
	is,
	\[u_a(\pi(a))=u_b(\pi(a))=\nu T+\nu T+T,\]
	while they have different values for agent~$b$'s bundle and
	\begin{align*}
	 & u_a(\pi(b))=2 \big( \nu T+(\nu-\km)T+(T-S) \big) \quad \text{and} \\
	 & u_b(\pi(b))=(\nu-\km)T+(T-S).
	\end{align*}
	($\Rightarrow$)
	If the instance of \kSum is a Yes-instance, then we can delete the
	resources from $\pi(a)$ corresponding to a solution to the \kSum instance and get a new allocation $\pi'$,
	where we
	have that
	$\pi'(b)=\pi(b)$ and that
	\[u_a(\pi'(a))=u_b(\pi'(a))=\nu T+(\nu-\km)T+(T-S).\]
	Since $r_1^*$ is still in $\pi'(a)$, for agent~$b$ we have
	\[u_b(\pi'(a))-u_b(r_1^*)=(\nu-\km)T+(T-S)=u_b(\pi'(b)).\]
	For agent~$a$ we have
	\[u_a(\pi'(b))-u_a(r_2^*)=\nu T+(\nu-\km)T+(T-S)=u_a(\pi'(a)).\]
	Thus $\pi'$ is an EF1 allocation.

	($\Leftarrow$)
	If the instance of \EFODGN is a Yes-instance, then we show that the instance of
	\kSum is a Yes-instance.
	Let $P$ be the set of up to~$\km $ resources we delete.
	Because $u_b (\pi (a)) - u_b (\pi (b)) = \km T + \nu T +S $ and every set~$R' \subseteq \pi (a)\setminus \{r_1^*\}$ of at most~$\km -1$ resources from~$\pi (a)\setminus \{r_1^*\}$ has $u_b (R') < \km T$, set~$P$ has to contain at least
	$\km$ resources from $\pi(a)$ in order to guarantee EF1.
	Thus, since our budget is $\km$, set~$P$ does not contain resources
	from
	$\pi(b)$ and $P$ contains exactly $\km$ resources from $\pi(a)$.
	Moreover, $P$ does not contain $r_1^*$, as otherwise $u_a (\pi (b) \setminus \{r_2^*\}) = u_a (\pi (b) \setminus \{r_3^*\}) = \nu T + (\nu-\km ) T + (T-S) > (\nu + 1) T = u_a (\pi (a) \setminus \{r_1^*\}) \ge u_a (\pi (a) \setminus P)$ and therefore $a$ envies~$b$ up to one good.
	Since~$a$ does not envy~$b$ up to one good, we have
	\[u_a(\pi(a))-u_a(P) \ge u_a(r_2^*) \Rightarrow u_a(P) \le \km T+S,\]
	and since~$b$ does not envy~$a$ up to one good, we have
	\[u_b(\pi(b)) \ge u_a(\pi(a))-u_a(P)-u_a(r_1^*) \Rightarrow u_a(P) \ge \km
	T+S.\]
	Therefore, $u_a(P)=\km T+S$ and hence the instance of \kSum is a
	Yes-instance.

	For $\kp$, we can use the same reduction with setting
	$\kp=k+3$
	and applying the following modification
	\[u_a(r_2^*)=u_a(r_3^*)=\nu T+k T+S \quad \text{and} \quad
	u_b(r_2^*)=u_b(r_3^*)=\frac{k T+S}{2}.\]
	It is easy to verify that the instance of \EFODGN is a Yes-instance if and
	only if we can keep exactly $k$ resources from $\pi(a) \setminus \{r_1^*\}$
	with
	value $kT+S$ and $r_1^*$ for agent $a$, which means that the instance of
	\kSum
	is a Yes-instance.
\end{proof}

\subsection{EF and Identical Valuations} \label{sub:EFID}
We now turn to EF and start by proving a strong hardness result for
\EFDGNW:
Reducing from \textsc{Subset Sum}, we show that even if we only have two
agents with identical valuations,
\EFDGNW
is NP-hard. Even stronger, \EFDGNW is NP-hard even if we are allowed to
delete all but one resource.

\begin{restatable}{theorem}{ibkmkp}\label{ib:kmkp}
	For identical binary valuations, \EFDGN with $\kp=1$ and \EFDGW with $\lp=1$
	are NP-hard for $n=2$ agents, and \EFDGN with $n=2$ agents is W[1]-hard
	parameterized by $\km$.
\end{restatable}
\begin{proof}
	For the NP-hardness of \EFDGN, we present a reduction from \subsetSum.
	Let $x_1, x_2, \ldots, x_\nu$, and $S$ be an instance of \subsetSum.
	We create an instance of \EFDGN with two agents $a$ and $b$ with identical valuations.
	In the initial allocation~$\pi$ the agent $a$ has $\nu$ resources---one for
	each
	$x_i$ which $a$ and $b$ value as $x_i$.
	Agent~$b$ receives in~$\pi$ a single resource with utility value $S$.
	We set $\kp=1$.
	Suppose the instance of \subsetSum is a Yes-instance and let $I$ be the
	solution.
	Then, the agent~$a$ can delete the resources in $I$ and we get an envy-free
	allocation with at least $\kp=1$ resources.
	If the instance of \subsetSum is a No-instance, then it is impossible for
	agent $a$ to hold a bundle of value $S$.
	Since agent $b$ only has one resource of value $S$, the only envy-free
	allocation	is the empty allocation.
  Thus the instance of \EFDGN is a No-instance.
	The same reduction with the modification that $\lp=1$ shows the NP-hardness
	of \EFDGW.

	For the W[1]-hardness of \EFDGN, we present a reduction from \kmSum (as defined in \Cref{sub:prob}).
  	Let $x_1, x_2, \ldots, x_\nu$, $k$, and $S$ be an instance of \kmSum.
  	We create an instance the same as above and set $\km =
  	k$.
  	Suppose the instance of \kmSum is a Yes-instance and let $I$ be the
solution.
  	Then, the agent~$a$ can delete the resources in $I$ and we get an envy-free
  	allocation.
  	Notice that $|I| \le k =\km$.
  	If the instance of \kmSum is a No-instance, then in any envy-free
  	allocation the utility for two agents is either~0, which corresponds to the
  	empty allocation, or $S$, which needs more than $\km$ resources to be
  	deleted from agent $b$'s bundle.
  	Thus, the instance of \EFDGN is also a No-instance.
\end{proof}

This strong hardness result puts EF1 and EF in a sharp contrast, as \EFODGN is
polynomial-time solvable and \EFODG is solvable in time polynomial in $\lp$.
Concerning the later contrast, on a high level, the
reason why $\lp$ is a more powerful
parameter for \EFODG than for \EFDG is that we know by
\Cref{ib:EF1-least} that there always is an EF1 allocation where all agents
hold a bundle of utility at least $\min_{a \in A} u(\pi(a))$. Thus, for
``small'' values of $\lp$ we can return yes, while ``larger'' values of
$\lp$ allow for a tractable algorithm in $\lp$.

However, again similar to the EF1 case, using a simplified version of Step 1
from \Cref{ib:EF1-lp}, in the
following theorem, we show that in contrast to the W[1]-hardness with respect
to $\km$ and NP-hardness for $\lp=1$ from the previous theorem,
\EFDGW (and, in fact, even \EFDG) is solvable in
time polynomial in $\lm$.

\begin{restatable}{theorem}{EFgeneral}\label{ib:EF-general}
	For identical valuations, \EFDG can be solved in
$\mathcal{O}((\lm)^5+|\mathcal{I}|)$
time.
\end{restatable}
\begin{proof}
	Let $t$ be the utility that each agent will have after the deletion of
	resources for its bundle. We first compute the range of all possible values
	of $t$.
	Let $t^*=\min_{a_i \in A} u(\pi(a_i))$ be the minimum utility among all agents
	in the initial allocation, then $t \le t^*$.
	Let $t_{\min}$ be the minimum integer number $t$ such that $\sum_{a_i \in A}
	(u(\pi(a_i))- t) \le \lm$.
	Since we can lose at most $\lm$ utility, we have $t \ge t_{\min}$, so $t
\in
[t_{\min},t^*]$.
	By the definition of $t_{\min}$ we have $t^*-t_{\min} \le \lm$, so we need
to
	guess at most $\lm$ different values of $t \in [t_{\min},t^*]$.

	Next for each guess of $t$, let $A(t) \subseteq A$ be the set of agents
with
utility larger than~$t$.
	Notice that $|A(t)| \le \lm$ as otherwise we need to delete more than $\lm$
	resources and the utility will be decreased by more than $\lm$.
	For every agent $a_i \in A(t)$, we need to find a subset $P_i \subseteq
	\pi(a_i)$ with the minimum size such that $u(P_i)=u(\pi(a_i))-t \le \lm$.
	Since $u(P_i) \le \lm$, subset~$P_i$ can only contain resources from $\pi(a_i)$ with
	utility at most $\lm$ and for each such value at most $\lm$ resources.
	Thus, we can find a subset $S_i \subseteq \pi(a_i)$ with $|S_i| \le
	(\lm)^2$ such that we only need to include resources from $S_i$ in $P_i$.
	This is just an instance of the \minsubsetSum problem (as introduced in
	\Cref{sub:prob}) with integers from $\{u(r) \mid r \in S_i\}$ with $|S_i| \le (\lm)^2$ and
	target sum $u(\pi(a_i))-t \le \lm$, which can be solved in
	$\mathcal{O}((\lm)^3)$ time using dynamic programming.
  We need to do this for each $t \in [t_{\min},t^*]$
  and each agent $a_i \in A(t)$ separately.
  As we have $t^*-t_{\min} \le \lm$ and $|A(t)| \le \lm$,
  this can be done in $\mathcal{O}((\lm)^5)$ time.
  Let $s_i(t)$ be the minimum size.
	For each~$t$, we check whether $\sum_{a_i \in A(t)} s_i(t) \le \km$.
	If this holds for at least one $t \in [t_{\min},t^*]$, then our instance of \EFDG is a Yes-instance.
	Otherwise, it is a No-instance.

	For the running time, we need $\mathcal{O}(|\mathcal{I}|)$ time to compute the range
$[t_{\min},t^*]$.
	So the overall running time is $\mathcal{O}((\lm)^5+|\mathcal{I}|)$.
\end{proof}

\section{Unary Valuations} \label{sub:uv}

In this section, we consider unary-encoded valuations and conduct a thorough
parameterized complexity analysis of our problems with respect to various
natural and problem-specific parameters. We start by proving that all our
problems are NP-hard even for 0/1-valuations by establishing a simple
reduction from
\textsc{Set Cover} to \EFDGNW/\EFODGNW:
\begin{restatable}{theorem}{unnp}\label{un:np}
	\EFDGNW and \EFODGNW are NP-hard for 0/1-valuations.
\end{restatable}
\begin{proof}
	We first prove the theorem for \EFDGN and later describe how to adapt
	the construction the other three problems.

	We reduce from \textsc{Set Cover}.
	Given an instance~$(S,\mathcal{C},z)$ of \textsc{Set Cover}, for
	each element~${s\in S}$, let $f_s$
	denote the number of sets from $\mathcal{C}$ in which $s$ appears.
	We construct an
	instance of \EFDGN by adding an \emph{element agent} $a_s$ for each
	$s\in S$. Moreover, we create a \emph{special agent} $a^*$.
	For each $s\in S$, we introduce $f_s-1$ resources that are valued as one by
	agent
	$a_s$ and valued as zero by every other agent. We allocate these resources
	to
	$a_s$ in the initial allocation $\pi$.

	Moreover, for each set $C\in \mathcal{C}$, we introduce a resource $r_C$
	that is valued as one by all element agents corresponding to elements from
	$C$,
	i.e., $\{a_s \mid s\in C\}$ and valued as zero by every other agent. We
	allocate all these resources to the
	special agent~$a^*$ in the initial allocation $\pi$. We set~$\km:=z$. In the
	initial
	allocation, the value that each element agent has for its own bundle is one
	lower
	than
	its value for the special agent's bundle. Thus, the allocation resulting
from deleting a set~$R'$ of at most $\km$ resources from the special agent's bundle is
envy-free if and only if each element agent
	values at least one resource from $R'$ as one, which is equivalent to each
	element appearing in one
of the sets corresponding to resources from $R'$. (For \EFDGW, we need to
slightly adapt the
	construction and let the special
	agent value the resources assigned to it in the initial
	allocation as~one and set $\lm:=z$.)

	\smallskip

	To adapt the construction for \EFODGN, we first of all assume without loss
	of generality that in the given \textsc{Set Cover} instance each element
	appears in
	at least two sets (if an element appears in only one set, then
	the
	set has to appear in every solution and we can thus delete the set and all
	elements
	occurring in it from the instance). Concerning the construction, in
	the initial allocation,
	for each $s\in S$, we now only allocate $f_s-2$ resources that are
	valued as one by $a_s$ and valued as zero by everyone else to $a_s$.
	Thus,
	in the
	initial
	allocation, the value each element agent has for its own bundle is two
	lower than
	its value for the special agent's bundle and we again need to delete at
	most
	$\km$ resources from the special agent's bundle such that each element agent
	values at least one of them as one. For \EFODGW, a similar adaption works.
\end{proof}

This hardness result motivates us to explore the parameterized complexity of
our problems.\footnote{Note that in the reduction from \Cref{un:np}, we set the number/welfare of the deleted resources
	$\km/\lm$ to be the requested size $z$ of the set cover. As
\textsc{Set Cover} is W[2]-hard parameterized by
	$z$, the reduction already establishes that \EFDGNW and \EFODGNW
	are
	W[2]-hard for
	$\km/\lm$.}
	 We start by considering in \Cref{sub:sparsebalanced} our problems in case
	 that valuations and the initial allocation are ``sparse''. Subsequently, in
\Cref{sub:n}, we consider the influence of the
	 number of agents on the problem.  Afterwards, in \Cref{sub:kmlm}, we turn
	 to the problem specific parameter~$\km$ and $\lm$, i.e., the maximum
	 number/welfare of the deleted resources. Lastly,
	 in \Cref{sub:kplp}, we consider the dual parameters $\kp$ and $\lp$ that
	 quantify the minimum number/welfare of remaining resources. We also
	 consider combined parameters, where we
	 mostly include results for a combined parameter in the latest section
	 regarding one of the parameters from the combination. We refer to
\Cref{tab:results_1} for an overview of results
from
\Cref{sub:sparsebalanced,sub:n}, to \Cref{tab:results_2} 
for an overview of results from \Cref{sub:kmlm}, and to \Cref{tab:results_3} 
for an overview of results from \Cref{sub:kplp}.

\subsection{Sparse Valuations and Allocations} \label{sub:sparsebalanced}
\begin{table}[t]
\caption{Overview of our results for parameters $n$, $\wa$, $\wrr$, and~$d$ and
their combinations. }
\label{tab:results_1}
\begin{center}
\resizebox{0.9\textwidth}{!}{\begin{tabular}{l|l|l}
\EFDGN\ & \EFDGW\  & \EFDG\ \\
\EFODGN  &
\EFODGW  & \EFODG \\ \midrule \midrule
NP-h. for $d=3$/$d=7$,& NP-h. for $\wrr=5$ and 0/1
val. (Co. \ref{vw}) &\\
$\wa=4$/$\wa=7$, $\wrr=3$, and 0/1-val. (Th. \ref{nkmw}) & FPT wrt. $d$,
FPT wrt. $\wa$ (Th. \ref{thm:welfaredwa})
&\\
\midrule
W[1]-h. wrt. $n+\km$ (Th. \ref{thm:w-hard-n+k-OG}) &W[1]-h. wrt. $n+\lm$ (Th.
\ref{thm:w-hard-n+k-OG})  &\\
& & FPT wrt. $u^*+n$ (Pr. \ref{nu}) \\
& & FPT wrt. $d+n$ (Ob. \ref{vn}) \\
& & FPT wrt. $\wa+n$ (Ob. \ref{vn}) \\

    \end{tabular}}
\end{center}
\end{table}
In this section, we consider instances where the initial allocation is sparse,
i.e., the
maximum number $d$ of resources assigned to an agent in the initial allocation
is bounded, and valuations are sparse, i.e., the maximum number $\wa$ of
resources that an agent values as non-zero and the maximum number $\wrr$ of
agents that value a specific resource as non-zero are bounded (see \Cref{tab:results_1} for an overview of our results from this and the next section.
We prove that \EFDGN and \EFODGN are NP-hard even for
constant value of $d+\wa+\wrr$ and 0/1-valuations. In a sharp contrast to
this, we show that \EFDGW and \EFODGW are fixed-parameter tractable with
respect to $\wa$ or $d$ for arbitrary valuations (and NP-hard for a constant
value of~$\wrr$ and 0/1-valuations).

First of all note that one can reduce \textsc{Restricted
Exact Cover by 3-Sets}, which is an NP-hard special case of \textsc{Set Cover}
where each set consists of exactly three elements and each
element appears in exactly three sets \cite[Problem
SP2]{DBLP:books/fm/GareyJ79} to our problems (by applying the reduction from \Cref{un:np} but
reducing from \textsc{Restricted
Exact Cover by 3-Sets} instead of \textsc{Set Cover}).
This allows us to establish the following:
\begin{corollary} \label{vw}
	\EFDGN/\EFODGN is NP-hard for 0/1-valuations even if $\wa=5$/$\wa=4$ and
$\wrr=3$. \EFDGW/\EFODGW is NP-hard for 0/1-valuations even if $\wrr=4$.
\end{corollary}

Indeed, using a more involved reduction, one can show that for
\EFDGN/\EFODGN this
para-NP hardness result can be extended to
also include the maximum number $d$ of resources an agent holds in the initial
allocation:
\begin{restatable}{theorem}{nkmw}\label{nkmw}
	\EFDGN is NP-hard for 0/1-valuations, even if $d=3$, $\wa=4$, and $\wrr=3$.
	\EFODGN is NP-hard for 0/1-valuations, even if $d=7$, $\wa=7$, and $\wrr=3$.
\end{restatable}
\begin{proof}

\begin{table*}[t]
    \caption{Valuations used in the proof of \autoref{nkmw}. To this end, let~$C \in \mathcal{C}$, and $s \in C$. 
    All valuations not covered by the table are zero.
    If a resource is contained in an agent's bundle, then the corresponding entry is written in boldface.
    The value of the initially allocated bundles are~$1$ for~$a(C)$, $1$~for~$a(C,\force,i)$ for~$i \in \{1,2\}$, $3$~for~$a(s)$ (owning three resources),
    and $2$~for~$a(s,\force)$.}
    \centering
    \resizebox{0.95\textwidth}{!}{\begin{tabular}{c c c c c c c c c c c c}
                     & $r(C,\select)$ & $r(C,\unselect,1)$ & $r(C,\unselect,2)$ & $r(C,\force,1)$ & $r(C,\force,2)$ & $r(s,C)$ & $r(s,\force,1)$ & $r(s,\force,2)$ \\
     $a(C)$          &      \tbf1     &      \tbf0         &      \tbf0         &      0          &        0        &     1    &       0         & 0   \\
     $a(C,\force,1)$ &      1         &      1             &      0             &      \tbf1      &        0        &     0    &       0         & 0   \\
     $a(C,\force,2)$ &      1         &      0             &      1             &      0          &    \tbf1        &     0    &       0         & 0   \\
     $a(s)$          &      0         &      0             &      0             &      0          &        0        & \tbf1    &       0         & 0   \\
     $a(s,\force)$   &      0         &      0             &      0             &      0          &        0        &     1    &   \tbf1         & \tbf1   \\
    \end{tabular}}
    \label{fig:nkmw}
\end{table*}

    We reduce from \textsc{Restricted Exact Cover by 3-Sets}.
	Given a collection~$\mathcal{C}$ of size-$3$ subsets of a set~$S$
	where each element appears in exactly three subsets from~$\mathcal{C}$,
	the question is whether	there exists an \emph{exact cover} $\mathcal{C}'\subseteq \mathcal{C}$,
	that is, a subcollection where each element $s\in S$ appears in exactly one set from $\mathcal{C}'$.
	We assume that $S$ contains~$3t$ elements for some~$t\in \mathbb{N}$ since otherwise there clearly cannot be an exact cover.
	Let $S=\{s_1,\dots,s_{3t}\}$ denote the set of elements and $\mathcal{C}=\{C_1,\dots,C_{3t}\}$
	denote the collection.
	Note that $\mathcal{C}'\subseteq \mathcal{C}$ is an exact cover if it covers each element and~$|\mathcal{C}'|=t$.

	We create the following agents:
	\begin{itemize}
	 \item for each subset~$C\in \mathcal{C}$ create three \emph{set agents}~$a(C)$, $a(C,\force,1)$ $a(C,\force,2)$, and
	 \item for each element~$s \in S$ two \emph{element agent}~$a(s)$ and $a(s,\force)$.
	\end{itemize}
    We have the following resources for each subset $C \in \mathcal{C}$:
	\begin{itemize}
	 \item $r(C,\select)$, $r(C,\unselect,1)$ and $r(C,\unselect,2)$, owned by $a(C)$, and
	 \item $r(C,\force,i)$ owned by $a(C,\force,i)$, for $i \in \{1,2\}$.
	\end{itemize}
    Furthermore, for each element~$s \in S$ and $C \in \mathcal{C}$ with $s\in C$ we add resource
	 $r(s,C)$, owned by~$a(s)$.
	 Finally, for each $s \in S$, we add resources $r(s,\force,1)$ and $r(s, \force, 2)$, owned by $a(s,\force)$.
    The valuations are described in \autoref{fig:nkmw}.
    The budget~$\km$ is set to~$8t$.
    This completes the construction.

    \paragraph{Idea and Basic Observation.}
    The element agents for an element $s \in S$ are designed in a way that~$a(s,\force)$
    envies~$a(s)$ enforcing it to delete at least one resource~$r(s,C)$.
    Deleting more than one resource from~$a(s)$ would result in exceeding the budget~$\km$.
    In essence, deleting~$r(s,C)$ will correspond to subset~$C$ being part of the exact cover
    (and element~$s$ being covered by~$C$).
    The set agents for a set~$C \in \mathcal{C}$ are designed such that~$a(C,\force,1)$ and $a(C,\force,2)$
    both envy~$a(C)$ enforcing it to either delete resource~$r(C,\select)$ (which will correspond
    to~$C$ being part of the cover) or to delete both resources~$r(C,\unselect,1)$ and~$r(C,\unselect,2)$.
    Deleting the resource~$r(C,\select)$ requires that all resources~$r(s,C)$ for each~$s \in C$
    are deleted as well since otherwise set agent~$a(C)$ envies element agent~$a(s)$.

    \paragraph{Correctness.}
    We show that there exists an exact cover $\mathcal{C}'\subseteq \mathcal{C}$ for the given instance
    if and only if one can delete $8t$~resources to obtain an envy-free allocation.

    For the ``only if'' direction, assume there is an exact cover $\mathcal{C}'\subseteq \mathcal{C}$.
    Now, for each~$\{s,s',s''\}=C \in \mathcal{C}'$, delete the resource~$r(C,\select)$ from agent~$a(C)$
    and delete $r(s,C)$, $r(s',C)$, and $r(s'',C)$ from the agents~$a(s)$,
    $a(s')$, and $a(s'')$, respectively.
    Moreover, for each~$C \notin \mathcal{C}'$ delete~$r(C,-,1)$
and~$r(C,-,2)$. Note that, overall, we have deleted $8t$ resources.
    It is easy to see that for each~$C \in \mathcal{C}$ and $s \in S$ the agents~$a(C,\force,1 )$, $a(C, \force, 2)$, $a(s,\force)$,
    and $a(s)$ are not envious.
    Consider an agent~$a(C)$.
    If $C \notin \mathcal{C}'$, then~$a(C)$ cannot be envious because it still has resource~$r(C,\select)$ of value~$1$
    and each other agent even initially owned a bundle of value at most~$1$ from the perspective of agent~$a(C)$.
    If $C \in \mathcal{C}'$, then each other agent owns a bundle of value~$0$ from the perspective of agent~$a(C)$:
    the bundles of other set agents have value~$0$ anyway (from $a(C)$'s perspective) and the bundle of
    each element agent has now also value~$0$ since we deleted all resources~$r(s,C)$ with~$s \in C$.

    For the ``if'' direction, assume one can delete $8t$~resources to obtain an envy-free allocation.
    Clearly, due to agents~$a(C,\force)$ and $a(s,\force)$ one must have deleted at least one item from
    each set agent~$a(C)$ and one item from each element agent~$a(s)$.
    Moreover, for each set agent~$a(C)$ with~$C=\{s,s',s''\}$ from which one deletes only one resource, the deleted resource must be~$r(C, +)$ and
    one must delete three resources from element agents, namely from~$a(s)$, $a(s')$, and $a(s'')$.
    We now call a set agent~$a(C)$ from which one deletes only one resource a \emph{selected set agent};
    all other set agents are called \emph{unselected set agents}.
    We claim there are exactly~$t$ selected set agents.
    Assume towards a contradiction there are~$t-z, z\in \mathbb{N}$ such agents.
    Deleting at least~$3t$ resources from element agents, at least $4t+2z$
    resources from unselected set agents, plus the~$t-z$ resources from selected set agents
    together clearly exceeds our budget; a contradiction.
    Next, assume towards a contradiction there are~$t+z, z\in \mathbb{N}$ selected set agents~$a(C)$.
    As discussed, this requires at least~$3t+3z$ deletions from element agents.
    Moreover, one needs $4t-2z$~resources from unselected set agents and
    $t+z$~deletions from selected set agent.
    This implies altogether~$8t+2z>k^-$ deletions in total; a contradiction.
    Finally, it is easy to verify that the $t$~sets corresponding to the~$t$
    selected set agents form an exact cover:
    an element~$s \in S$ not being covered would imply an envious element agent~$a(s,\force)$ and by simple counting arguments
    one can see an element cannot be covered twice.

    \paragraph{Sketch for \EFODGN.}
    To show the result for \EFODGN, we modify the above reduction by creating
    for each $C \in \mathcal{C}$ two resources~$r(C, 1)$ and $r(C, 2)$ in the bundle of~$a (C)$ where for $i\in \{1, 2\}$, resource~$r(C, i)$ is valued as~1 by $a(C, \force, i)$ and as 0 by every other agent.
    Furthermore, for every~$ s \in S$, we add a resource $r(s, 1)$ in the bundle of~$a(s)$ where~$r(s, 1)$ is valued as~1 by $a(s, \force)$ and as 0 by every other agent.
    Finally, for every~$C\in \mathcal{C}$ and every~$s\in C$, we add a resource~$r(s, C, 1)$ in the bundle of~$a(s)$ which is valued as~1 by~$a(C)$ and as 0 by every other agent.
\end{proof}

In contrast to this, we now show that \EFDGW and \EFODGW are fixed-parameter
tractable with respect to $d$ or $\wa$ for arbitrary (even binary) valuations.
This result
is surprising
in two ways. First, we have proven in the previous theorem that
\EFDGN and \EFODGN are NP-hard even for a constant value of $d+\wa+\wrr$ and
0/1-valuations.
Notably, this is our first and only case for unary valuations where one of
\EFDGN and \EFDGW (or \EFODGN and \EFODGW) is harder than the other.
Second, we have already seen
that \EFDGW and \EFODGW are NP-hard for a constant value of the
closely
related parameter $\wrr$ for 0/1-valuations.
\begin{restatable}{theorem}{welfaredwa}\label{thm:welfaredwa}
	\EFDGW/\EFODGW are solvable in $\mathcal{O}(2^{2d} \cdot d \cdot n^2 +
|\mathcal{I}|)$ time and solvable in
$\mathcal{O}(2^{2\wa} \cdot \wa \cdot n^2 + |\mathcal{I}|)$ time.
\end{restatable}
\begin{proof}
We first describe our approach for \EFDGW.
First, we delete all resources that are valued as zero by the agent holding
them
in the initial allocation and denote the resulting allocation by $\pi_0$.
Next, as long as there is an agent $a$ that is envied by some agent, we
compute a subset $S_a \subseteq \pi_0(a)$ of resources with the highest
utility for $a$ such that $a$ will not be envied after we delete
all resources in $\pi_0(a) \setminus S_a$. Subsequently, we set $a$'s bundle to
be $S_a$ (but do not update $\pi_0(a)$).
Some agent may become envied by other agents multiple times over the execution of the algorithm. Each time, we recompute the
optimal subset $S_a$ from $\pi_0(a)$. This implies, for instance, that an agent
might get back once deleted~resources.

Let $\pi^*$ be the resulting allocation.
It is clear that $\pi^*$ is envy-free.
To show that $\pi^*$ has the highest utilitarian welfare among all envy-free
allocations, we show an even stronger result.
We say an allocation (which is not necessarily envy-free) is \emph{maximal} if
every agent has higher or equal utility in this allocation than in every
envy-free allocation that can result from deleting resources from the initial
allocation $\pi$.
We show that $\pi^*$ is maximal by induction.
First of all, $\pi_0$ is maximal, as every agent has the same utility as in the
initial allocation.
Next, we show that if in some step the current allocation $\pi_1$ is maximal,
then after this step the
resulting allocation $\pi_2$ is also maximal.
Let $a$ be the agent whose bundle has been changed in this step.
Since $\pi_1$ is maximal and $a$ is the only agent with $\pi_1(a) \neq
\pi_2(a)$, it suffices to show that for each envy-free allocation $\pi'$,
we have $u_a(\pi_2(a)) \ge u_a(\pi'(a))$.
Suppose for contradiction that $u_a(\pi_2(a)) < u_a(\pi'(a))$. Then, we
modify $\pi_2$ by giving agent $a$ the bundle $S'_a:=\pi'(a) \cap \pi_0(a)$ and
get a new allocation $\pi'_2$ (where by our above assumption agent $a$
has a higher utility in $\pi'_2$ than in $\pi_2$).
We claim that in $\pi'_2$ agent $a$ is not envied by other agents.
Indeed, since $\pi_1$ is maximal and every agent $a' \in A\setminus\{a\}$ has
the same set of resources in $\pi'_2$, $\pi_2$, and $\pi_1$, we have that agent
$a'$ has higher or equal utility in $\pi'_2$ as in $\pi'$, i.e.,
$u_{a'}(\pi'_2(a'))=u_{a'}(\pi_2(a'))=u_{a'}(\pi_1(a')) \ge u_{a'}(\pi'(a))$.
Moreover, since $\pi'$ is envy-free, which means that $a$ is not envied by other
agents in $\pi'$, and $\pi'_2(a)=S'_a \subseteq \pi'(a)$, we get that $a$ is not
envied by other agents in $\pi'_2$.
However, this contradicts the choice of~$\pi_2(a)$ because $S_a'$ is a set of higher utility than~$\pi_2(a)$ such that no agent envies~$a$.
Thus, $\pi^*$ is maximal.
 For the running time, note that in $\pi_0$ every agent holds at most
 $d':=\min\{\wa,d\}$ resources.
 In each step, we can iterate, in $\mathcal{O}(2^{d'}\cdot
 d' \cdot n)$ time, over all subsets $S$ of the bundle of the selected
 agent $a$ and check whether deleting resources $\pi_0(a)\setminus S$ resolves
 the envy of all agents towards $a$.
 In each step, we update the bundle for some agent because this bundle is envied
 by other agents. Since the utility of every agent is non-increasing during
 the whole process, a bundle cannot reappear in a later step.
 Thus, we can bound the number of executed steps by the sum of the number of
 different
 subsets of~$\pi_0(a)$ for all $a \in A$, which is $2^{d'} \cdot n$.
 Thus, the running time of the algorithm is $\mathcal{O}(2^{2d} \cdot d \cdot
 n^2 + |\mathcal{I}|)$ or
 $\mathcal{O}(2^{2\wa} \cdot \wa \cdot n^2 + |\mathcal{I}|)$.
 
 For \EFODGW, we use the same algorithm with the modification that as long as
 there is an agent $a$ that is envied up to one good by other agents, we compute
 a subset $S_a \subseteq \pi_0(a)$ of resources with the highest value such that
 $a$ will not be envied up to one good after we delete all resources in $\pi_0(a)
 \setminus S_a$.
 In the proof of optimality, we define maximal by comparing with
 all EF1
 allocations.

  For the running time, note that in $\pi_0$ every agent holds at most
 $d':=\min\{\wa,d\}$ resources.
 In each step, we can iterate, in $\mathcal{O}(2^{d'}\cdot
 d' \cdot n)$ time, over all subsets $S$ of the bundle of the selected
 agent $a$ and check whether deleting resources $\pi_0(a)\setminus S$ resolves
 the envy of all agents towards $a$.
 In each step, we update the bundle for some agent because this bundle is envied
 by other agents. Since the utility of every agent is non-increasing during
 the whole process, a bundle cannot reappear in a later step.
 Thus, we can bound the number of executed steps by the sum of the number of
 different
 subsets of~$\pi_0(a)$ for all $a \in A$, which is $2^{d'} \cdot n$.
 Thus, the running time of the algorithm is $\mathcal{O}(2^{2d} \cdot d \cdot
 n^2 + |\mathcal{I}|)$ or
 $\mathcal{O}(2^{2\wa} \cdot \wa \cdot n^2 + |\mathcal{I}|)$.
 
 For \EFODGW, we use the same algorithm with the modification that as long as
 there is an agent $a$ that is envied up to one good by other agents, we compute
 a subset $S_a \subseteq \pi_0(a)$ of resources with the highest value such that
 $a$ will not be envied up to one good after we delete all resources in $\pi_0(a)
 \setminus S_a$.
 In the proof of optimality, we define maximal by comparing with
 all EF1
 allocations.
\end{proof}

\subsection{Number of Agents} \label{sub:n}

Having seen that at least for \EFDGN and \EFODGN sparse valuations and
allocations do not help
in our search for tractable cases, we now turn to the number $n$ of agents.
Before
we identify two tractable cases, we
start by showing that our problems are W[1]-hard with respect to the number
of agents even considered in combination with the number (welfare) of deleted
resources (we will examine the parameters number of welfare of deleted resources in more detail in the next subsection).

\begin{restatable}{theorem}{EFOnkl}\label{thm:w-hard-n+k-OG}
\label{thm:w-hard-n+k}
	For unary valuations, \EFDGNW and \EFODGNW are W[1]-hard
	parameterized by $n+\km$/$n+\lm$.
\end{restatable}
\begin{proof}
   We start by proving the theorem for EF and later describe how to adapt it
for EF1.
    We reduce from the MRSS problem:
    \decprob{\textsc{Multidimensional Relaxed Subset Sum} (MRSS)}
{An integer $k$, a set $S= \{s_1, \dots, s_{\nu}\} \subseteq \mathbb{N}^k$, a
target
vector~$t \in \mathbb{N}^k$, and an integer $k'$.}
{Is there a $S' \subseteq S$ with $|S'| \le k'$ and $\sum_{s \in S'} s
\ge t$?}
MRSS is W[1]-hard parameterized by $k + k'$ even if all numbers are encoded
in unary~\cite{DBLP:journals/algorithmica/GanianKO21}.
	For $j\in [k]$ and $i\in [\nu]$, we denote the $j$-th entry of vector $s_i$
	by $s_i[j]$.
	Given an instance of MRSS, we construct an instance of
	\EFDGNW as follows.
	We assume that~$k' \le \nu$ (if $k' > \nu$, then reducing $k'$ to $\nu$
clearly results in an equivalent instance).
	The agents are $a_1, \dots, a_k, a^*$.
	Furthermore, there are $\nu+ k$ resources, namely $r_1, \dots, r_k$ and
	$r_1^*, \dots, r_{\nu}^*$.
	In the initial allocation $\pi$, for $i\in [k]$, resource $r_i$ is
	allocated
	to agent~$a_i$
	and resources $r_1^*,
	\dots, r_v^*$ are allocated to $a^*$.
	Agent~$a^*$ values resource $r_i^*$ as 1 for~$i\in [\nu]$, and all
	other resources as
	0.
	For $j\in [k]$ and $i\in [\nu]$, agent~$a_j$ values resource $r_i^*$ as
	$s_i[j]$, resource $r_j$ as
	$\sum_{i=1}^\nu s_i[j] - t[j]$, and, for $j'\neq j$, resource $r_{j'}$ as 0.
	We set $\km:= k'$ or $\lm := k' $.

	($\Rightarrow$)
	Let $S' =\{s_{i_1}, \dots, s_{i_{k'}}\}$ be a solution to MRSS.
	We delete resources $r_{i_1}^*, \dots, r_{i_{k'}}^*$, and claim that this
	results in an envy-free allocation.
	Clearly, we delete at most $\km$ resources and the welfare of the deleted
	resources is at most
	$\lm$ as $a^*$ values each of the at most $\km$ deleted resources as 1.
	Agent $a^*$ clearly does not envy any agent.
	For $j\in [k]$, agent $a_j$ values its own resource as $\sum_{i=1}^\nu
	s_i[j] - t[j]$, and
	the resources of $a^*$ as $\sum_{s \in S\setminus S'} s[j] =
	\sum_{i=1}^\nu
	s_i[j] - \sum_{s\in S'} s[j] \le \sum_{i=1}^n s_i[j] - t[j]$ and thus does
	not envy any other agent.

	($\Leftarrow$)
	Let $R'$ be a set of at most $\km$ resources/resources of welfare at most
	$\lm$ such that their deletion results in an
	envy-free allocation.
	We may assume without loss of generality that $R'$ does not contain a
	resource $r_j$ for $j\in [k]$, as deleting such a resource can never reduce
	the envy as $r_j$ is only valued as non-zero by the agent holding it in
	$\pi$ and thus removing such a resource from $R'$ results
	in a smaller
	solution.
	Furthermore, we may assume that $R'$ contains exactly $k'$ resources
	(otherwise we can add arbitrary resources $r^*_i$ to $R$).
	Thus, we have that $R' = \{r^*_{i_1}, \dots, r^*_{i_{k'}}\}$.
	We claim that $S':= \{s_{i_1} , \dots, s_{i_{k'}}\}$ is a solution to the
	MRSS instance.
	Since for all $j\in [k]$, $a_j$ does not envy $a^*$, we have that
	$\sum_{i=1}^\nu s_i[j] - t[j]
	\ge \sum_{r\in R\setminus R'} u_{a_j} (r) = \sum_{s \in S\setminus S'}
	s[j]$.
	Thus, we have $\sum_{s\in S'} s[j] \ge t[j]$, and therefore, $S'$ is a
	solution to the MRSS instance.

    {\bf Adaption to EF1:}
	We adopt the reduction for EF by adding for every $j
	\in [k]$ a resource~$\hat r_j$, which is allocated to $a^*$ and valued
	$\sum_{s\in S} s[j]$ by $a_j$, valued $ 1$ by~$a^*$, and 0 by all
	other agents.

	The forward direction works identically; to show that agent $a_{j}$ does
	not envy $a^*$ up to one resource, the resource $\hat r_j$ will be the one
	resource from $a^*$ that $a_j$ ignores.

	To see the backward direction, first note that we may assume without loss
	of generality that the solution~$R'$ does not contain a resource $\hat r_j$
	(if it did, then we could find another resource by replacing $\hat r_j$ by
	the resource which is best-valued by $a_j$ but not contained in $R'$; this
	does not change whether $a_j$ envies $a^*$. Moreover, the exchange can
only decrease
	the envy of agent $a_{j'}$ for $j' \neq j$ because $a_{j'}$ values $\hat
	r_j$ as 0, and $a^*$ does not envy any other agent because it evaluates all
	resources not allocated to $a^*$ as 0).
	The rest of the backward direction is analogous to the case of EF.
\end{proof}

On the positive side, combining the number of agents with
the maximum utility value an agent assigns
to a resource, our general problems can be encoded in an ILP where the number
of constraints is quadratic in $n$. Subsequently, we can employ the algorithm by
	Eisenbrand and Weismantel \cite{DBLP:journals/talg/EisenbrandW20}:

\begin{restatable}{proposition}{nuu}\label{nu}
	\EFDG and \EFODG are
	solvable in
	$(n^2u^*)^{\mathcal{O}(n^2)}\cdot m^2$ time.
\end{restatable}
\begin{proof}
We start by considering envy-freeness.
	For
	each resource $r\in R$, we introduce a binary variable~$x_{r}$ which is $1$
	if we do not delete resource $r$ and $0$ otherwise. The given instance of
	\EFDG
	admits a solution if the following ILP admits a feasible solution:
	\begin{align}
	&\sum_{r\in \pi(a)} u_a(r)\cdot x_r \geq \sum_{r\in \pi(a')} u_{a}(r)\cdot
	x_r, \qquad  \forall (a,a')\in A\times A
	\label{cond:ef}\\
	& \qquad \qquad \sum_{a\in A}\sum_{r\in \pi(a)} u_a(r)\cdot x_r \geq  \lp,
\qquad \quad \quad
	\sum_{r\in R} x_r \geq \kp
	\label{cond:deletion}
	\end{align}
	Condition \ref{cond:ef} enforces that no agent envies another agent.
	The first part of  Condition \ref{cond:deletion} enforces
	that the allocation after the deletion of resources has utilitarian
	welfare at least $\lp$, while the second part enforces that at least $\kp$
resources
	remain. As we have $n^2+2$ constraints and the maximum value in the
	constraint matrix is upper bounded by~$u^*$, applying the algorithm by
	Eisenbrand and Weismantel \cite{DBLP:journals/talg/EisenbrandW20}, the ILP can
	be created and solved in
	$(n^2u^*)^{\mathcal{O}(n^2)}\cdot m^2$ time.

	For envy-freeness up to one good a slightly more involved approach is
	needed. For each agent pair $(a,a')\in A\times A$, we guess an integer
	$y_{a,a'}\in [0,u^*]$ which denotes the maximum value that $a$ assigns to a
	undeleted resource from $a'$'s bundle. Note that we need to guess overall
	$\mathcal{O}((u^*)^{n^2})$ different cases. The ILP is then adjusted as follows:
	\begin{align}
	y_{a,a'}+\sum_{r\in \pi(a)} u_a(r)\cdot x_r \geq \sum_{r\in \pi(a')}
	u_a(r)\cdot
	x_r, \qquad & \forall (a,a')\in A\times A
	\label{cond:ef1}\\
	\sum_{r\in \pi(a'): u_a(r)=y_{a,a'}} x_r \geq 1, \qquad & \forall (a,a')\in
	A\times A
	\label{cond:guess}\\
	\sum_{r\in \pi(a)} u_a(r)\cdot x_r \geq  \lp, \qquad & \forall a\in A
	\label{cond:util1}\\
	\sum_{r\in R} x_r \geq \kp
	\label{cond:deletion1}
	\end{align}
	The ILP differs from the ILP for EF in two places. First of all when
	comparing the utility value of $a$ for $a$'s bundle with the utility of $a$
	for
	$a'$'s bundle, we add $y_{a,a'}$ to the value of $a$ for $a$'s bundle as we
	have assumed that $y_{a,a'}$ is the utility of $a$ for an undeleted
	resource in
	the bundle of $a'$. That at least one resource $r$ in $a'$'s bundle with
	$u_a(r)=y_{a,a'}$ remains is enforced by Condition \ref{cond:guess}.
\end{proof}

The previous result implies that our general problem for both EF and EF1
is in FPT with respect to $u^*+n$, in XP with respect
to
$n$, and FPT with respect to $n$ for 0/1-valuations.

Examining now the combination of the number $n$ of agents with the parameters
introduced in
\Cref{sub:sparsebalanced}, first note that $n$ upper-bounds
the number $\wrr$ of
agents that value a resource as non-zero. However, if we combine the number $n$ of
agents with the maximum number $\wa$ of resources that an agents values as
non-zero or the
maximum number $d$ of resources an agent holds in the initial allocation, then
we can bound the number of ``relevant'' resources and thereby
the size of the whole instance in a function of the combined parameter~$n + \wa$ or $n + d$:

\begin{restatable}{observation}{vn}\label{vn}
	\EFDG and \EFODG are
	solvable in $\mathcal{O}(|\mathcal{I}|+n^2\cdot m \cdot  2^{n\cdot \wa})$
	time
and solvable in
	$\mathcal{O}(|\mathcal{I}|+n^2 \cdot  m \cdot 2^{n\cdot d})$ time.
\end{restatable}
\begin{proof}
 	Note that as each of the $n$ agents only values at most $\wa$ resources as
	non-zero, there can exist at most $n\cdot \wa$ resources that are valued as
	non-zero by some agent. As deleting a
	resources that is
	valued as zero by everyone never makes a difference, it suffices to guess
	for
	each of the
	$n\cdot \wa$ resources that are valued as non-zero by some agent whether the
	resource is deleted or not and check for each guess whether the
	resulting allocation is envy-free (up to one good).

	For the parameter combination $n$ and $d$ this is even more obvious, as the
	total number of resources is bounded by $n\cdot d$.
\end{proof}

\subsection{Number/Welfare of Deleted Resources} \label{sub:kmlm}

\begin{table}[t]
\caption{Overview of our results for parameters $\km$ and $\lm$ and all
parameter combinations they are involved in.}
\label{tab:results_2}
\begin{center}
\resizebox{0.9\textwidth}{!}{\begin{tabular}{l|l|l}
\EFDGN\ & \EFDGW\  & \EFDG\ \\
\EFODGN  &
\EFODGW  & \EFODG \\ \midrule \midrule
& XP wrt. $\lm$ & XP wrt. $\km$ \\
 &  & NP-h. for $\wa=5$/$\wa=4$, \\
  &  &  $\wrr=3$, $\lm=0$, and 0/1-val. (Co. \ref{pr:lmgneral}) \\
 W[2]-h. wrt. $\km$  & W[2]-h. wrt. $\lm$  & \\
 for 0/1-val. (Co. \ref{km}) &for 0/1-val.
 (Co. \ref{km}) & \\\midrule
W[1]-h. wrt. $n+\km$ (Th. \ref{thm:w-hard-n+k-OG}) &W[1]-h. wrt. $n+\lm$ (Th.
\ref{thm:w-hard-n+k-OG})  &\\ \midrule
& FPT wrt. $d$ (Th. \ref{thm:welfaredwa})  & FPT wrt. $d+\km$ (Pr. \ref{vkm})
 \\
&   & XP wrt. $d+\lm$ (Pr. \ref{EFDG:lmd})  \\
&   & W[1]-h. wrt. $\lm$  \\
&   & for $d=4$/$d=5$ (Pr. \ref{EFDG:lmd})  \\
& FPT wrt. $\wa$ (Th. \ref{thm:welfaredwa}) & FPT wrt. $\wa+\km$ (Pr.
\ref{vkm})  \\
& FPT wrt. $u^*+\wrr+\lm$ (Pr. \ref{kmw}) & FPT wrt. $u^*+\wrr+\km$ (Pr.
\ref{kmw})

    \end{tabular}}
\end{center}
\end{table}
In this section, we examine the influence of the number/welfare~$\km$/$\lm$ of resources to be
deleted and identify some tractable cases (see \Cref{tab:results_2} for an
overview of our results).
On the hardness side, in \Cref{thm:w-hard-n+k,thm:w-hard-n+k-OG}, we
have already shown that this parameter (even in combination with the number of
agents) is not sufficient to lead fixed-parameter tractability (unless FPT=W[1]).
Moreover, recall
that in \Cref{un:np} we have constructed a (parameterized) reduction from the
W[2]-hard \textsc{Set Cover} problem to our problems which also establishes the
following:

\begin{corollary} \label{km}
	Parameterized by $\km$/$\lm$, \EFDGNW and \EFODGNW are W[2]-hard for
	0/1-valuations.
\end{corollary}

On the algorithmic
side, there is an XP algorithm with respect to parameter $\km$ for \EFDG
and \EFODG
running in time $\mathcal{O}(|\mathcal{I}|+m^{\km}\cdot \km\cdot n^2)$ by
simply iterating over all size-$\km$ subsets of resources and checking whether
the allocation that results from deleting these resources is envy-free (up to
one good).
For the parameter $\lm$, there also is an XP algorithm running in
$\mathcal{O}(|\mathcal{I}|+m^{\lm}\cdot \lm\cdot n^2)$ for \EFDGW and \EFODGW
by first deleting
all resources that are valued as zero by the agent holding it (which never
creates any envy) and subsequently iterating over all size-$\lm$
subsets of resources and checking whether additionally deleting these resources
makes the allocation envy-free (up to one good). However, in the first step of
the XP algorithm for $\lm$, an arbitrary number of resources might get deleted
which raises
the
question whether there is an XP algorithm for the parameter $\lm$ for the
general problems.
In fact, using the construction from \Cref{un:np} for \EFDGN/\EFODGN and
reducing from
\textsc{Restricted Exact Cover by 3-Sets} while setting
$\km=z$ and $\lm=0$, rules out this possibility:

\begin{corollary} \label{pr:lmgneral}
 \EFDG/\EFODG is NP-hard for 0/1-valuations even if $\lm=0$, $\wa=5$/$\wa=4$,
and
 $\wrr=3$.
\end{corollary}

We now consider $\km/\lm$ in combination with sparse valuations or a sparse
initial allocation and design several fixed-parameter tractable algorithms.
Using a search-tree approach, we show fixed-parameter tractability of
the parameter combination $d+\km/d+\lm$ and $\wa+\km/\wa+\lm$ for arbitrary
(even binary) valuations (notably, we have
already proven in \Cref{thm:welfaredwa} that \EFDGW/\EFODGW is fixed parameter
tractable in $d$ or $\wa$):

\begin{restatable}{proposition}{vkm}\label{vkm}
	\EFDG/\EFODG is solvable in  $\mathcal{O}(|\mathcal{I}|+  n\cdot
	d^{\km})$ time and solvable in $\mathcal{O}(|\mathcal{I}|+  n\cdot
	\wa^{\km})$ time .
	\EFDGW/\EFODGW is solvable in $\mathcal{O}(|\mathcal{I}|+  n\cdot
	d^{\lm})$ time and solvable in $\mathcal{O}(|\mathcal{I}|+  n\cdot
\wa^{\lm})$ time.
\end{restatable}
\begin{proof}
	We solve these problems using a simple search-tree approach and start by
	considering the algorithms for $\km$:
	Until
	the allocation $\pi$ is envy-free (up to one good) or the budget is zero,
	we
	check whether there still exists an agent $a\in A$ which envies another
	agent
	$a'\in A$ (up to one good). In this case we branch over deleting one
	resource
	that
	$a$ values as non-zero from the bundle of $a'$ in $\pi$. In case the
	budget
	becomes zero and $\pi$ is not envy-free (up to one good), we reject the
	current
	branch.

	The correctness of the algorithm is immediate so it remains to consider its
	running time. Initially, we can check which agents envy each other (up to
	one
	good) in $\mathcal{O}(|\mathcal{I}|)$. After deleting a resource from $a'$'s
	bundle, we can update this information in $\mathcal{O}(n)$ (as we
	only
	need to consider agent pairs involving $a'$). The depth of the recursion is
	bounded by $\km$. Moreover, in case we can bound the size of each agent's
	bundle by $d$, the branching factor is at most $d$, while in case we can
	bound
	the number of resources an agent values as non-zero by $\wa$, we can bound
	the
	branching factor by $\wa$. In all cases, the running times stated above
	follow.

	For $\lm$, we start by deleting all resources that are valued as zero by
	the
	agent holding the resource in $\pi$. By doing so, we can only reduce envy
	and
	do not change the budget. After this, we can delete at most $\lm$
	additional
	resources, as each remaining resource is valued as non-zero by the agent
	holding it. Thus, we can again employ the search-tree approach described
	above
	to solve the problem.
	Notably, this algorithm for $\lm$ does not solve the general problem as in
	its
	first step an arbitrary number of resources might get deleted thereby
	violating a possible bound $\km$ on the number of resources that got
	deleted.
\end{proof}
As already observed above for the XP algorithm, here again $\lm$ in contrast to
$\km$ is
not enough to establish an algorithmic result for the general problem. Observe
that we have already seen in \Cref{pr:lmgneral} that both general problems are
NP-hard for constant $\wa$ and constant $\lm$.
Moreover,
we show in the following proposition using a slightly involved reduction from
\textsc{Clique} that the FPT algorithm for $d+\lm$
cannot be extended to the general problem (like for $d+\km$). However,
\EFDG/\EFODG
parameterized by $d+\lm$ is in XP for arbitrary (binary or unary) valuations.

\begin{restatable}{proposition}{lmd}\label{EFDG:lmd}
	\EFDG/\EFODG is solvable in $\mathcal{O}(|\mathcal{I}|+m^{\lm}\cdot 2^d
\cdot m
\cdot  n)$ time.
For unary valuations, parameterized by $\lm$, \EFDG/\EFODG is
	W[1]-hard even if $d=4$/$d=5$.
\end{restatable}
\begin{proof}
	We always consider EF first and afterwards describe how our approach for EF
	can be adapted for EF1.

	\textbf{XP algorithm for ${d+\lm}$:} We start by guessing the set of
	resources to delete that are valued as non-zero by the agent holding it
	(there can exist at most $\lm$ of them) and delete them from the initial
	allocation. Subsequently, we know that we are only allowed to delete
	resources that are valued as zero by the agent holding it (notably we
	cannot create new envy by deleting such resources). Thus, we can solve the
	problem of which further resources to delete for each of the agents
	separately: For each agent $a\in A$, we iterate over all subsets of
	resources the agent holds and values as zero. For each such subset, we
	check whether deleting it resolves the envy of all agents toward $a$. If no
	such subset exists, we reject the current guess; otherwise we delete the
	smallest such subset (breaking ties arbitrarily). After executing this
	procedure for all agents, we check whether we deleted at most $\km$
	resources. If this is the case, then we return yes and otherwise we reject
	the current guess.  For EF1, the same algorithm where we check for each
	subset of
	resources an agent holds and values as zero whether deleting this
	subset of resources resolves the envy up to one good of all other agents
	towards this agent.
	\medskip

	\textbf{W[1]-hardness for $\lm$ for constant $d$:} We show the hardness
	result by a reduction from \textsc{Clique}.
	\textsc{Clique} asks, given a graph~$G = (V, E)$ and an integer~$t$, whether $G$ contains a clique of size~$t$, that is, a set of $t$ pairwise adjacent vertices.
	Given an instance of
	\textsc{Clique} consisting of a graph $G=(V,E)$ and an integer $t$ (we
assume that $t \ge 2$ as otherwise the instance is a trivial yes instance), we
	construct an instance of \EFDG as follows.  We
	insert a \emph{vertex agent} $a_v$ for each vertex $v\in V$ and an
	\emph{edge agent} $a_e$ for each edge $e\in E$. Each vertex agent $a_v$
	holds two resources in the initial allocation $\pi$: One resource~$r_v^1$
	which is valued as $3$ by itself and as $0$ by everyone else and one
	resource~$r_v^2$ which is valued as $0$ by itself and as $1$ by all edge
	agents that correspond to edges incident to~$v$. For each edge
	$e=\{u,v\}\in E$, agent~$a_e$ holds four resources in $\pi$: Three resources
	$r_e^1, r_e^2, r_e^3$ that are valued as $0$ by itself and as $1$ by the
	two vertex agents $a_u$ and $a_v$. Moreover, $a_e$ holds a resource $r_e^4$
	that it values as $1$ and that is valued as $3$ by  $a_u$ and
	$a_v$. We set $\lm={{t}\choose{2}}$ and $\km=3|E|-2{{t}\choose{2}}+t$. In
	summary, all vertex agents value
	their bundle as $3$ and the bundles of all incident edge agents as $6$
	(and all other bundles as $0$).
	Each edge agent values its own bundle as $1$ and the bundles of vertex
	agents corresponding to its endpoints as $1$ (and all other bundles as
	$0$).

	$(\Rightarrow)$ Assume we are given a clique $V'\subseteq V$ of size $t$ in
	$G$. For an edge $e=\{u,v\}\in E$, we write $e\subseteq V'$ if $u\in V'$
	and $v\in V'$. From this, we construct a solution to the constructed \EFDG
	instance by deleting for each vertex $v\in V'$ resource $r_v^2$, for each
	edge $e\subseteq V'$, resource $r_e^4$ and for all edges $e\not\subseteq
	V'$, resources $r_e^1, r_e^2,$ and $r_e^3$. As there exist
	${{t}\choose{2}}$ edges lying in the clique, the total number of resources
	we delete is $t+{{t}\choose{2}}+3(|E|-{{t}\choose{2}})=\km$. Moreover, the
	only resources that were deleted and that are valued as non-zero by the
	agent holding them are the ${{t}\choose{2}}$ resources $\{r_e^4 \mid
	e\subseteq V'\}$. As they are valued as one by $a_e$, the budget~$\lm$ is respected.
	In the
	resulting allocation, there is no envy of vertex agents towards edge agents
	corresponding to incident edges as resources of value three from the
	perspective of the vertex agent got deleted from the bundle of the edge
	agent. Concerning the utilities of the edge agents, all edge agents
	corresponding to edges outside the clique still value their bundle as one
	and are therefore not envious. Concerning edge agents corresponding to
	edges inside the clique, they value their bundle as zero. However, we also
	deleted all resources they value as non-zero from the bundle of other
	agents (namely, one resource from the bundle of each vertex agent
	corresponding to the edge's endpoint). Thus, the resulting allocation is
	envy-free.

	$(\Leftarrow)$ Assume that there exists a solution to the \EFDG instance which deletes a set~$R'\subseteq R$ of resources.
	Because every
	vertex agent envies each edge agent corresponding to an incident edge in
the initial allocation, for each edge $e\in E$, resources $r_e^1, r_e^2,$ and
	$r_e^3$ or resource~$r_e^4$ needs to be deleted. Let $\alpha$ be the number
	of edge agents for which resources $r_e^1, r_e^2,$ and $r_e^3$ are part of
	$R'$ and $\beta$ the number of edge agents for which resource $r_e^4$ is
	part of $R'$. It needs to hold that $\alpha+\beta\geq m$ and because of
	$\lm$ that $\beta\leq {{t}\choose{2}}$. Further, let $\gamma$ denote the
	number of vertex agents from whose bundle a resource got deleted. Notably,
	for each edge $e\in E$ for which $r_e^4$ is deleted, a resource from the
	vertex agents corresponding to the endpoints of $e$ needs to be deleted, as
	otherwise $a_e$ envies the respective vertex agent. Thus, it needs to hold
	that
	${{\gamma}\choose{2}}\geq \beta$. Thus, overall we have that:
	\begin{align*}
	& 3\alpha+\beta+\gamma\leq \km = 3|E|-2{{t}\choose{2}}+t  \\& \land
	\alpha+\beta \geq |E|   \land
	\beta \leq {{t}\choose{2}}   \land
	{{\gamma}\choose{2}}\geq \beta
	\end{align*}
	If $\gamma = t-x$ for $x\in \mathbb{N}$, then $\beta \le
\binom{t-x}{2}\le \binom{t}{2} - x$ (using $0 \le x \le t$ and $t\ge 2$ for the
last inequality).
	Therefore, we have $\alpha \ge |E|- \binom{t}{2} + x$.
	Consequently, we have $3 \alpha + \beta + \gamma \ge 3|E| - 3\binom{t}{2} +
3x + \binom{t}{2} -x + t-x = 3|E| - 2\binom{t}{2} + t + x$. As $3 \alpha + \beta
+ \gamma\leq \km=3|E| - 2\binom{t}{2} + t$,  this implies $x = 0$ and thus,
$\gamma = t $ and $\beta = \binom{t}{2}$.
	Therefore, there are $\binom{t}{2}$ edges (namely those for which $r_e^4$
	is deleted) whose endpoints are $t$ vertices, implying that these $t$
	vertices form a clique.

	\medskip

	For EF1, we slightly modify the construction: For each vertex $v\in V$,
	vertex agent $r_v$ now holds a resource $r_v^1$ that is valued as three by
	itself and as zero by everyone else and two resources $r_v^2$ and $r_v^3$
	that are valued as zero by itself and as two by all edge agents
	corresponding to edges incident to $v$. Turning the the edge agents, for
	each edge
	$e=\{u,v\}\in E$, agent~$a_e$ holds five resources in $\pi$: Three
	resources
	$r_e^1, r_e^2, r_e^3$ that are valued as $0$ by itself and as $1$ by the
	two vertex agents $a_u$ and $a_v$. Moreover, $a_e$ holds two resources
	$r_e^4$ and $r_e^5$
	that it values as $1$ and that is valued as $3$ by  $a_u$ and
	$a_v$.
\end{proof}

Lastly, we turn to the combined parameter $\km$ plus the maximum number of
agents $\wrr$ that value a resource as non-zero. As the number $n$ of agents
upper-bounds $\wrr$, as proven in
\Cref{thm:w-hard-n+k},
this parameter combination is not enough to achieve fixed-parameter
tractability. However, by adding the maximum utility value, tractability can be
regained:
\begin{restatable}{proposition}{kmw}\label{kmw}
	\EFDG and \EFODG are solvable in
$\mathcal{O}(({u^*}+1)^{{\km}^2\cdot
	(\wrr+1)}\cdot \km \cdot |\mathcal{I}|)$ time. \EFDGW and \EFODGW are
solvable in
$\mathcal{O}(({u^*}+1)^{{\lm}^2\cdot
	(\wrr+1)}\cdot \lm \cdot |\mathcal{I}|)$ time.
\end{restatable}
\begin{proof}
For a subset of agents $A'\subseteq A$ and a resource $r$, we call
$(u_a(r))_{a\in A'}$ the utility profile of $r$ restricted to $A'$.
We start by considering the parameter combination $\km+u^*+\wrr$ and solve \EFDG
and \EFODG by a similar approach but start with \EFDG. Our
algorithm crucially relies on the simple observation that by deleting $\km$
resources the envy of at most
$\km\cdot \wrr$ agents can be resolved, as for each resource $r\in R$ only
$\wrr$
agents
value $r$ as non-zero. Assuming that we are given a valid solution where we
delete
at most $\km$ resources one after each other, let $A'\subseteq A$ be the set
of agents from which a resource got delete or that envy another agent at some
point. From our above observation it follows that $|A'|\leq
\km \cdot (\wrr+1)$. Now assuming that we would knew $A'$ we could guess for
each
of the $\km$ resources $r\in R$ to be delete the utility profile of $r$
restricted to $A'$ (there exist ${u^*}^{{\km}^2\cdot (\wrr+1)}$ guesses) and
subsequently delete a resource matching the utility profile. However,
unfortunately, we do not
know the set of agents $A'$ upfront which makes it necessary to update this set
of agents on the fly and complete the previous guesses concerning the utility
profiles of deleted resources. We do so in \Cref{alg}, where we bookmark the
set $E$ of agents that either envy another agent at some point during the
course of the algorithm or from whose bundle a resource got deleted (Lines
\ref{line:init-E} and \ref{line:extendE}). Note that only agents from~$E$ can envy
other agents after the deletion of resources, as all agents that are initially
envious in $\pi$ are part of~$E$ and after deleting some resources only the
agents to which at least one of the resources belong can become additionally
envious (and they are also part of $E$).
As long as there exists an agent $a\in E$
envying another agent $a'$, which implies that a resource from $a'$'s bundle
needs to be deleted, we guess the utility which agents from~$E$ have for the
resource to be deleted from $a'$'s bundle (Line \ref{line:guessp}) and adjust
the utility that these agents have for $a'$'s bundle (Line~\ref{line:incorpp}).
Moreover, if $a'$ was not already part of $E$, we add $a'$ to $E$ and guess the
value that $a'$ has for all ``deleted'' resources and adjust $a'$'s valuations
accordingly (Line~\ref{line:extendp}). Finally, in Lines
\ref{line:checkS}-\ref{line:checkE} we check whether there in fact exists a set
of at most $\km$ resources satisfying our guesses.

\begin{algorithm}[t!]
		\begin{algorithmic}[1]
		\State Let $E$ be the set of agents that are envious in the initial
allocation $\pi$. \label{line:init-E}
		\If{$|E|>\km\cdot \wrr$} \label{line:reject-E}
Reject
            \EndIf
            \State $P:=\{\}$
			\While{there is an agent $a\in E$ whose current value for its own
bundle is smaller than its current value for the bundle of another agent $a'\in
A$}
            \If{$a'\notin E$}
            \For{$(\tilde{p},\tilde{a})\in P$}
            \State Extend $\tilde{p}$ by guessing the value $\tilde{p}_{a'}\in
[0,u^*]$
that $a'$ has for the respective resource and
decrease the value that $a'$ has for $\tilde{a}$'s bundle by $\tilde{p}_{a'}$.
\label{line:extendp}
            \EndFor
            \State Set $E:=E\cup \{a'\}$ \label{line:extendE}
            \EndIf
			\State Guess the utility profile $p=(p_a)_{a\in E}\in [0,u^*]^{|E|}$
of
the resource to be deleted
from $a'$'s bundle restricted to $E$. \label{line:guessp}
\For{$a\in E$} Decrease the value that $a$ has for $a'$'s bundle by $p_a$.
\label{line:incorpp}
\EndFor
            \State Add $(p,a')$ to $P$.
            \State $\km=\km-1$.
            \If{$\km<0$} Reject \EndIf
			\EndWhile

			\For{$(\tilde{p},\tilde{a})\in P$} \label{line:checkS}
			\If{$\tilde{a}$ does not hold a resource that matches utility
profile $\tilde{p}$} Reject \label{line:rejectp}
\EndIf
        \State Delete a resource matching utility
profile $\tilde{p}$ from $\tilde{a}$'s
bundle. \label{line:checkE}
			\EndFor
			\If{the resulting allocation has utilitarian welfare at
least $\lp$} Accept. \Else { } Reject \EndIf

			\end{algorithmic}
		\caption{\label{alg}  Algorithm for \EFDG parameterized by
$\km+\wrr+u^*$}
	\end{algorithm}

Concerning the running time of the algorithm, note that $E$ has size at most
$\km \cdot (\wrr+1)$ and that we execute the while-loop at most $\km$ times.
Thus, in the end we guess at most $\km$ utility profiles of length at most $\km
\cdot (\wrr+1)$ containing entries from $[0,u^*]$. Overall, there exist
$({u^*}+1)^{{\km}^2\cdot
	(\wrr+1)}$ such
guesses.
For each such guess, we need to update the utility value of all agents from $E$
with $|E|\leq \km \cdot (\wrr+1)$
and check whether they envy another agent and in the end whether there exists
a subset of resources matching the guesses. All this can be done in
$\mathcal{O}(\km \cdot |\mathcal{I}|)$.

Let us now turn to proving the correctness of the algorithm. We start by
assuming that the algorithm returns Accept. In this case, we know that there
exists a subset $R'$ of resources  whose deletion ensures that no agent from $E$
is envious. As already observed above $E$ contains all agents that are
initially envious in $\pi$ and all agents that hold a resource from $R'$.
It follows that agents from $A\setminus E$ are not envious in $\pi$ and that
no resource from their bundle gets deleted. Thus, they are also not envious
in the resulting allocation.

Now assume that there exists a subset $R'\subseteq R$ of resources with
$|R'|\leq \km$ whose deletion makes $\pi$ envy-free. Then, deleting the
resources $R'$ can only resolve the envy of $\km \cdot \wrr$ agents. Thus, the
algorithm does not return Reject in \Cref{line:reject-E}. From $R'$ it is easy
to construct an accepting run of \Cref{alg}: With $E$ being the set of agents
that are envious in $\pi$, we start by selecting an agent
$a\in E$ that envies another agent $a'$.
There needs to exist a resource $r\in \pi(a')$ with $r\in R'$ as otherwise
after deleting $R'$, $a$ still envies $a'$. We now set the guess of
$(p_a)_{a\in E}$ to $(u_a(r))_{a\in E}$. If the utility profile $(p_a)_{a\in
E}$
gets extended at some later point because an agent $\tilde{a}$ is added to $E$,
we again set $p_{\tilde{a}}=u_{\tilde{a}}(r)$.
By repeating this procedure and also extending utility profiles in this way, as
$R'$ is a valid solution, the while-loop terminates after at most $\km$
iterations. Moreover, as $R'$ satisfies the current guesses, the algorithm also
does not reject in Line \ref{line:rejectp} and thus returns Accept.

By replacing ``envious'' by ``envious up to one good'' the described algorithm
with minor adjustments
and the proof also works for \EFODG.

Turning to the parameter combination $\lm+\wrr+u^*$, we start by deleting all
resources that are valued as zero by the agent holding it. Afterwards, as each
remaining resources is valued as non-zero by the agent holding it, at most
$\lm$ additional resources can be deleted. Thus, we can employ the algorithm
from above to solve \EFDGW and \EFODGW, where we now instead of decreasing
$\km$ by one for each resource we delete, we need to decrease $\lm$ by the
value assigned to the resource by the agent holding it in $\pi$.
\end{proof}
Notably, \Cref{pr:lmgneral} rules out the possibility that the FPT-algorithm
from above for \EFDGW and \EFODGW for parameters $\lm+u^*+\wrr$ can be extended
to
the respective general problem.

\subsection{Number/Welfare of Remaining Resources} \label{sub:kplp}
\begin{table}[t]
\caption{Overview of our results for parameters $\kp$ and $\lp$ and
parameter combinations involving them. }
\label{tab:results_3}
\begin{center}
\resizebox{0.9\textwidth}{!}{\begin{tabular}{l|l|l}
\EFDGNW & \EFODGN & \EFODGW\\  \midrule \midrule
NP-h. for $\kp=1/\lp=1$    & W[1]-h. wrt.
$\kp$ for $\wa=2$ & W[1]-h. wrt $\lp$ \\
 and $\wrr=4$ (Th. \ref{kp}) &  and 0/1 val. (Pr. \ref{di:okp}) & for
0/1 val. (Pr. \ref{di:okp}) \\ \midrule
FPT wrt. $\kp/\lp$ for 0/1-val. (Pr. \ref{di:kpp}) & \\ \midrule
& FPT wrt. $\kp+d$ (Pr. \ref{EF1:kpd}) &  FPT wrt. $\lp+d$ (Pr. \ref{EF1:kpd})

    \end{tabular}}
\end{center}
\end{table}

In this section, we examine the influence of the number/welfare $\kp$/$\lp$ of
resources that
remain after the deletion of resources (see \Cref{tab:results_3} for an
overview of our results). We
show that, generally speaking, these parameters are less powerful than the
respective dual parameters but we nevertheless identify two tractable cases. We
also show an interesting contrast between EF and EF1.

We start by proving a very strong intractability result for EF,
namely, that deciding
whether there is a solution that does not delete all resources is already
NP-hard:

\begin{restatable}{theorem}{akp}\label{kp}
	For unary valuations, \EFDGNW is NP-hard even
	if $\kp=1$/$\lp=1$ and
	$\wrr=4$.
\end{restatable}
\begin{proof}
	We reduce from \textsc{Independent Set} on cubic graphs, i.e., graphs
	where each vertex has exactly three neighbors. Given a cubic graph
	$G=(V,E=\{e_1,\dots, e_{q}\})$ and an integer $t$, the task is to decide whether $G$ contains an independent set of size~$t$. We assume without loss
	of generality that there are no isolated vertices in $G$ (as in this case
	we can add all isolated vertices to the independent set and decrease $t$
	accordingly) and we construct an instance
	of \EFDGNW as follows. We start by adding a \emph{special agent} $a^*$ and
	for each $i\in [q]$ an \emph{edge agent} $a_{e_i}$. Moreover, for each
	$v\in V$, we add a \emph{vertex resource} $r_v$, which is valued as~$1$ by
	the special agent and by all edge agents corresponding to edges incident to
	$v$ (note that as $G$ is cubic only three edges are incident to
	$v$). We allocate all vertex resources to the special agent. Furthermore,
	for $i\in [q]$, we add a resource $r_{e_i}$ which we allocate to $a_{e_i}$
	and which is valued as~$1$ by $a_{e_i}$ and by $a_{e_{i-1\bmod q}}$ and as
	$t$
	by the special agent. Lastly, we set $\kp=1$/$\lp=1$.

	($\Rightarrow$)
	Given an independent set $V'\subseteq V$ in $G$, we construct a solution to
	the constructed \EFDGNW instance by deleting the vertex resources
	corresponding to
	vertices which are not part of $V'$, i.e., $\{r_v\mid v\in V\setminus
	V'\}$. In the resulting allocation, the special agent has value $t$ for its
	bundle and the bundle of all edge agents. Each edge agent has value~$1$ for
	its bundle and at most value~$1$ for the bundle of any other edge agent.
Moreover, no edge agent  $a_{e_i}$ for some $i\in [q]$ can envy the special
agent, as this would imply that  $a_{e_i}$ has value at least two for the
special
	agent's bundle. From this it follows that both vertex resources
	corresponding to $e_i$'s endpoints are still part of the special agent's
	bundle, contradicting that $V'$ is an independent set.

	($\Leftarrow$) Assume we are given a subset $R'\subseteq R$ of resources
whose deletion results in the allocation
$\pi'$ which is a solution to  the
	constructed \EFDGNW instance. We claim that in $\pi'$ each edge agent
still needs to
	hold its resource and that the special agent needs to hold at least
	$t$ resources:
	If $\pi' $ assigns at least one resource~$r_v$ to the special agent, then
$\pi'$ needs to assign every edge agent~$a_e$ with~$ e$ incident to~$v$ the
resource~$r_e$, as otherwise $a_e$ envies $a^*$.
	Since $\pi '$ is not the empty allocation, it follows that $\pi'$ assigns
resource~$r_{e_i}$ to~$a_{e_i}$ for some~$i \in [q]$.
	Then $\pi'$ also needs to assign
	$a_{e_{i-1\bmod q}}$ resource~$r_{e_{i-1 \bmod q}}$, as
	otherwise $a_{e_{i-1\bmod q}}$ envies~$a_{e_i}$.
	By iterating this argument, it follows that all
edge agents need to
	hold their resource in $\pi'$. Moreover, as the special agents
values the
	resource which $a_{e_i}$ holds as $t$, the special agent needs to hold at
least
	$t$ vertex resources in $\pi'$.
    Thus, we have proven the claim.

	Let $V'=\{v\in V\mid r_v\in \pi'(a^*) \}$  be the set of vertices
	corresponding to all vertex resources that the special agent holds in
	$\pi'$. By our claim from above, it holds that $|V'|\geq t$. Then,
	$V'$ needs to form an independent set in $G$:
	if there exists
	an edge~$e_i\in E$ with both endpoints in $V'$, then the corresponding edge
agent
	$a_{e_i}$ values the special agent's bundle as two in $\pi'$ and is thus
	envious, a contradiction.
\end{proof}

Note that the above reduction crucially relies on the fact that the special
agent values some resources as $t$. In fact, if we have 0/1-valuations, we
can establish fixed-parameter tractability for the number and welfare (but
 not the general) problem by first deleting all agents with zero
welfare in the initial allocation and all resources they value as non-zero and
then making a case distinction based on whether $\kp/\lp$ is larger than
the remaining number of agents:

\begin{restatable}{proposition}{kpp}\label{di:kpp}
	\EFDGNW is solvable in
$|\mathcal{I}|+({\kp}^2)^{\mathcal{O}({\kp}^2)}\cdot
m^2$/$|\mathcal{I}|+({\lp}^2)^{\mathcal{O}({\lp}^2)}\cdot m^2$ time for
0/1-valuations.
\end{restatable}
\begin{proof}
        We solve the problem by applying the following algorithm and start by
considering the problem for $\kp$:
		As long as there exists an agent which does not hold a resource it
values as~$1$, we
		delete this agent together with all resources the agent values as~$1$
(as
it is clear that this agent will have utility zero in the final allocation and
thus that all resources the agent values as~$1$ need to be deleted as well).
		Thus, we can assume that every agent holds at least one resource in
the initial allocation in the altered instance.  Let $A'$ be the agents
remaining in the altered instance and $n'=|A'|$ the number of remaining
agents. Now, we make a case distinction based on the relationship between $n'$
and $\kp$.
		If $\kp < n'$,
		then we can easily construct a solution
with $n'>\kp$ remaining resources by deleting all resources already
deleted above and for all agents $a\in A'$ all but one resource that $a$
values as~$1$ from $a$'s bundle. In the resulting allocation $\pi'$ all agents
from
$A\setminus A'$ value their bundle and the bundle of everyone else as~$0$ and
all
agents from $A'$ value their bundle as~$1$ and the bundle of everyone, which
consists of at most one resource, as at most~$1$. Otherwise, if $\kp\geq n'$, we
employ the FPT-algorithm for the number of agents from \Cref{nu} in the altered
instance to solve the problem in $({\kp}^2)^{\mathcal{O}({\kp}^2)}\cdot m^2$
time.

For $\lp$, the same approach where we replace $\kp$ by $\lp$ works, as for the
second part of the algorithm in case that $n'<\lp$ the constructed allocation
has also utilitarian social welfare at $\lp$ as each agent holds at least one
resource it values as~$1$.
\end{proof}

In contrast to this result for EF, \EFODGNW parameterized by
$\kp/\lp$
is W[1]-hard for
0/1-valuations. We use
a reduction similar to the one used in \Cref{kp} but reduce from
\textsc{Independent Set} and delete all resources held by edge agents and set
$\kp=t/\lp=t$.

\begin{restatable}{proposition}{okp}\label{di:okp}
	Parameterized by $\kp$, \EFODGN is
	W[1]-hard for 0/1-valuations, even if $\wa=2$ and the
	initial allocation has zero utilitarian welfare.
	Parameterized by $\lp$, \EFODGW is
	W[1]-hard for 0/1-valuations.
\end{restatable}
\begin{proof}
	We start by proving the statement for \EFODGN by a reduction from
	\textsc{Independent Set} which is W[1]-hard parameterized by the size of
	the
	solution~\cite{DBLP:journals/tcs/DowneyF95}. Given a graph $G=(V,E)$ and an integer $t$, we construct an
	instance of \EFODGN as follows. For each edge $e\in E$, we add an
	\emph{edge
		agent} $a_e$. Moreover, we introduce a \emph{special agent} $a^*$.
		Lastly, we
	add for each vertex $v\in V$ a \emph{vertex resource}~$r_v$ that is valued
	as
	one by all edge agents corresponding to edges that are incident to $v$. In
	the
	initial allocation $\pi$, we allocate all vertex resources to the special
	agent. We set $\kp=t$. As the special agent values all resources as~$0$ and
	each edge agent values only the two vertex resources corresponding to its
	endpoints as~$1$, each agent values at most two resources as~$1$ and the initial
	allocation has zero utilitarian welfare.

	($\Rightarrow$)
	Given an independent set $V'\subseteq V$ in $G$, we construct a solution to
	the constructed \EFDGN instance by deleting the vertex resources
	corresponding to
	vertices which are not part of $V'$, i.e., $\{r_v\mid v\in V\setminus
	V'\}$. In the resulting allocation, the special agent holds $t$ resources.
	As each edge agent only values resources corresponding to one of
	its endpoints as~$1$ and $V'$ is an independent set, all edge agents only
value
	one
	resource from the bundle of the special agent in the resulting allocation as~$1$.
	Thus, the resulting allocation is envy-free up to one good.

	($\Leftarrow$) Let $R'\subseteq R$ be the set of at least $t$ resources
	allocated to the special agent in an envy-free allocation $\pi'$. We claim
	that
	the set of vertices $V'$ corresponding to the vertex resources from $R'$
	form
	an independent set in $G$. If there exist $v,v'\in V'$ with $e=\{v,v'\}\in
	E$,
	then agent $a_e$ values two resources from $R'$ as~$1$ and thus envies the
	special agent up to one good in $\pi'$, a contradiction.

	For \EFODGW, the reduction from above needs to be only slightly adapted by
	making the special agent value each vertex resources as~$1$ and setting
$\lp=t$.

\end{proof}

It is an intriguing open question whether \EFODGNW
is in XP with respect to $\kp/\lp$ or para-NP-hard. Note that an XP algorithm
for
unary valuations
would be
quite surprising, as this would imply that for the parameter $\kp/\lp$, for
unary valuations, EF1 is easier than EF, while, for 0/1-valuations,
EF is easier than EF1.

Lastly, by making a case distinction
based on whether more or less than $\kp/\lp$ agents hold a resource in the
initial
allocation, we show that \EFODGNW
is fixed-parameter tractable with respect to $\kp+d/\lp+d$ for arbitrary (unary
or binary) valuations:

\begin{restatable}{proposition}{EFOkpd}\label{EF1:kpd}
\EFODGNW is solvable in $\mathcal{O}(|\mathcal{I}|+2^{\kp\cdot d}
\cdot m \cdot  n^2)/ \mathcal{O}(|\mathcal{I}|+2^{\lp\cdot d}
\cdot m \cdot n^2)$ time.
\end{restatable}
\begin{proof}
We start by considering the parameter combination $\kp+d$.
 Let $n^*$ be the number of agents which hold at least one resource in the
 initial
allocation $\pi$.
 If $\kp \le n^*$, then deleting all but one resource of every agent results
clearly in an
envy-free up to one good allocation with $n^*\ge \kp$ allocated resources.
 Otherwise,
if $\kp > n^*$, then $\kp\cdot d > n^* \cdot d>m$ is larger than the
total number of resources. In this case, we iterate over all subsets
of at least~$\kp$ resources and check whether they form an envy-free up to one
good allocation. As there exist at most $2^m<2^{\kp\cdot d}$ subsets of
resources, the stated running time follows.

We now turn to the parameter combination $\lp+d$. Here we start by deleting all
resources that are valued as~$0$ by the agent holding it. Subsequently, let
$n^*$ be the number of agents that hold at least one resource (which they need
to
value as non-zero). Replacing $\kp$ by $\lp$, the rest of the algorithm works
as above, where the allocation constructed in the case $\kp \le n^*$ has
utilitarian welfare at least $\kp$, as each agent only holds resources it
values as non-zero.
\end{proof}

While our picture for the other parameters is nearly complete, there exist
various open questions for parameters $\kp$ and $\lp$. For instance, the
complexity of our problems parameterized by $n+\kp$, $n+\lp$, $d+\kp$, $d+\lp$, 
$\wa+\kp$, and $\wa+\lp$ is open.

\section{Conclusion}
We studied the complexity of making an initial allocation
envy-free by donating a subset of resources satisfying certain constraints.
While we have shown that this problem is NP-hard even under quite severe
restrictions on the input, resorting to parameterized complexity theory, we
identified numerous tractable cases. Moreover, we
discovered several interesting contrasts between seemingly closely related
parameters and problem variants.
For future work, it would be possible to change
the considered fairness criterion and to consider, for example,
proportionality or maximin share.
Furthermore, instead of or in addition to allowing that some resources are donated,
one could also analyze reallocating some resources,
which seems to be a harder setting
as many of our intractability results (such as Theorem \ref{un:np})
can be converted to this setting by adding some dummy agents to ``receive'' donated resources.
Lastly, from a more practical point
of view, it would be interesting to experimentally measure how many resources would need to be
donated to make allocations envy-free that are computed using an algorithm or
heuristic which tries to compute an allocation with few envy.

\section*{Acknowledgments}
 NB was supported by the DFG project MaMu (NI 369/19) and by the DFG project ComSoc-MPMS (NI 369/22).
 KH was supported by the DFG Research Training Group 2434 ``Facets of Complexity'' and by the DFG project FPTinP (NI 369/16).
 DK acknowledges the support of the OP VVV MEYS funded project CZ.02.1.01/0.0/0.0/16\_019/0000765 ``Research Center for Informatics''.
 JL was supported by the DFG project “AFFA” (BR~5207/1; NI~369/15) and the Ministry of Education, Singapore, under its Academic Research Fund Tier 2 (MOE2019-T2-1-045).

\end{document}